\documentclass[reqno,12pt]{amsart}
\usepackage{amsmath,amsthm,amscd,amssymb,latexsym,upref,stmaryrd}
\usepackage{amsfonts,mathrsfs,dsfont,bm}
\usepackage[top=1.5in, bottom=1.5in, left=1.4in, right=1.4in]{geometry}
\usepackage{color}
\usepackage{pdfsync}

\newcommand{\R}{{\bbR}}

\newcommand{\C}{{\mathbb C}}

\newcommand{\IR}{\mathbb{R}}
\newcommand{\ID}{\mathbb{D}}
\newcommand{\IC}{\mathbb{C}}

\newcommand{\oID}{\overset{\circ}{\mathbb{D}}}

\newcommand{\bbD}{{\mathbb{D}}}

\newcommand{\bbR}{{\mathbb{R}}}
\newcommand{\bbV}{{\mathbb{V}}}
\newcommand{\bbW}{{\mathbb{W}}}
\newcommand{\bbZ}{{\mathbb{Z}}}

\newcommand{\bsA}{{\boldsymbol{A}}}

\newcommand{\bsD}{{\boldsymbol{D}}}
\newcommand{\bsE}{{\boldsymbol{E}}}

\newcommand{\bsM}{{\boldsymbol{M}}}

\newcommand{\bsV}{{\boldsymbol{V}}}

\newcommand{\bsPsi}{{\boldsymbol\Psi}}

\newcommand{\bsPhin}{{\boldsymbol{\Phi_n}}}
\newcommand{\bsPhi}{{\boldsymbol\Phi}}

\newcommand{\bsAj}{{\boldsymbol{A^j}}}
\newcommand{\bsAk}{{\boldsymbol{A^k}}}
\newcommand{\bsAl}{{\boldsymbol{A^l}}}

\newcommand{\bssig}{{\boldsymbol{\sigma}}}

\newcommand{\bsPhint}{{\boldsymbol{\tilde\Phi_n}}}
\newcommand{\bsPhinp}{{\boldsymbol{\Phi_n^\prime}}}

\newcommand{\cA}{{\mathcal A}}
\newcommand{\cB}{{\mathcal B}}

\newcommand{\cD}{{\mathcal D}}

\newcommand{\cF}{{\mathcal F}}

\newcommand{\cH}{{\mathcal H}}

\newcommand{\cS}{{\mathcal S}}

\newcommand{\cV}{{\mathcal V}}

\newcommand{\bscA}{\boldsymbol{\mathcal{A}}}
\newcommand{\bscB}{\boldsymbol{\mathcal{B}}}

\newcommand{\dlab}{\big\langle\!\big\langle}
\newcommand{\drab}{\big\rangle\!\big\rangle}
\newcommand{\bdlab}{\bigg\langle\!\!\!\bigg\langle}
\newcommand{\bdrab}{\bigg\rangle\!\!\!\bigg\rangle}

\DeclareMathOperator{\supp}{supp}

 \DeclareMathOperator{\Tr}{Tr}

\renewcommand{\ln}{\text{\rm ln}}

\newcommand{\beq}{\begin{equation}}
\newcommand{\enq}{\end{equation}}

\newcommand{\Om}{\Omega}

\newcommand{\Omd}{\Omega_{_{disk}}}
\newcommand{\oOmd}{{\overset{\circ}{\Omega}}_{_{disk}}}

\newcommand{\hatt}{\widehat}

\renewcommand{\ge}{\geqslant}

\let\geq\geqslant
\let\leq\leqslant

\DeclareMathOperator{\diver}{div}

\makeatletter
\def\theequation{\@arabic\c@equation}

\allowdisplaybreaks 
\numberwithin{equation}{section}

\newtheorem{theorem}{Theorem}[section]
\newtheorem{proposition}[theorem]{Proposition}
\newtheorem{lemma}[theorem]{Lemma}
\newtheorem{corollary}[theorem]{Corollary}
\newtheorem{definition}[theorem]{Definition}
\theoremstyle{remark}
\newtheorem{remark}[theorem]{Remark}

\begin{document}

\title[Confinement of Dirac particles]{Essential self-adjointness of symmetric first-order differential systems 
and confinement of Dirac particles on bounded domains in $\mathbb R^d$}

\author[G.\ Nenciu]{Gheorghe Nenciu}
\address{Gheorghe Nenciu\\Institute of Mathematics ``Simion Stoilow'' of the Romanian Academy\\ 21, Calea Grivi\c tei\\010702-Bucharest, Sector 1\\Romania}
\email{Gheorghe.Nenciu@imar.ro}

\author[I.\ Nenciu]{Irina Nenciu}
\address{Irina Nenciu\\
         Department of Mathematics, Statistics and Computer Science\\ 
         University of Illinois at Chicago\\
         851 S. Morgan Street\\
         Chicago, IL \textit{and} Institute of Mathematics ``Simion Stoilow''
     of the Romanian Academy\\ 21, Calea Grivi\c tei\\010702-Bucharest, Sector 1\\Romania}
\email{nenciu@uic.edu}

\author[R.\ Obermeyer]{Ryan Obermeyer}
\address{Ryan Obermeyer\\
               Department of Mathematics, Statistics and Computer Science\\ 
         University of Illinois at Chicago\\
         851 S. Morgan Street\\
         Chicago, IL}
\email{roberm2@uic.edu}         

\thanks{The work of IN was partly supported by NSF grant DMS-1150427 and Simons Foundation grant 709025.
The work of RO was partly supported by NSF grants DMS-1150427 and DMS-1348092.}

\begin{abstract}
We prove essential self-adjointness 
of Dirac operators with Lorentz scalar potentials which grow sufficiently fast near the 
boundary $\partial\Omega$ of the spatial domain $\Omega\subset\mathbb R^d$.
On the way, we first consider general symmetric first order differential systems, for which
we identify a new, large class of potentials, called scalar potentials, ensuring
essential self-adjointness. Furthermore, using the supersymmetric structure of the Dirac operator in 
the two dimensional case, we prove confinement of Dirac particles, i.e.
essential self-adjointness of the operator, solely by magnetic fields $\mathcal{B}$
assumed to grow, near $\partial\Omega$, faster than $1/\big(2\text{dist} (x, \partial \Omega)^2\big)$.
\end{abstract}

\maketitle

\tableofcontents

\section{Introduction: the setting, the problem, and outline of results}\label{S:I}

The aim of this paper is to investigate the essential self-adjointess of Dirac operators, and other related 
first-order differential systems. Consider a domain (i.e. a connected open set) $\Omega\subset\bbR^d$, 
and on it
a first order, formally symmetric differential operator
\begin{equation}\label{E:I.1}
\ID=\bsE^{-1}\bigg(\sum_{j=1}^d \frac12\big(\bsAj D_j+D_j\bsAj\big)+\bsV\bigg)=\ID_0+\bbV\,,
\end{equation}
where $\bsE,\bsAj,\bsV$ are $k\times k$ matrix-valued functions and 
$D_j=-i\frac{\partial}{\partial x_j}$.
To focus ideas, we assume throughout the paper 
that $\bsE(x)>0$, $\bsAj(x)=\big(\bsAj(x)\big)^*$, and $\bsV(x)=\big(\bsV(x)\big)^*$ for all $x\in\Om$.
In addition, unless otherwise stated, we always take $\Omega$ bounded, and 
$\bsE,\bsAj\in C^1(\Om; \C^{k \times k})$, $\bsV\in C^0(\Om; \C^{k \times k})$ (see Comments~1 and~2 in Section~\ref{S:Comments} 
for a brief discussion of more general cases).
The operator $\ID$ is then symmetric on $\cD om(\ID)=C_0^1\big(\Om;\IC^{k}\big)$ with respect to the energy scalar product
\begin{equation*}
\dlab\bsPhi,\bsPsi\drab_{_\bsE}= \int_\Om \overline{\bsPhi(x)}\cdot
\big(\bsE\bsPsi\big)(x)\,dx\,,
\end{equation*}
and the problem is to find sufficient conditions on the coefficients $\bsE$, $\bsAj$, and $\bsV$ which ensure
that $\ID$ is essentially self-adjoint. 

In \cite{NN4}, the first two authors attacked this problem
by focusing on the principal symbol $\bsE^{-1}\sum \bsAj\xi_j$, generalizing and sharpening previous results. 
More precisely, let $\bsM(x)$ be the $d\times d$ real, non-negative definite (velocity) matrix given by
\begin{equation}\label{E:I.3}
\bsM_{jl}(x)=\text{Tr}\big(\bsE(x)^{-1/2}\bsAj(x)\bsE(x)^{-1}\bsAl(x)\bsE(x)^{-1/2}\big)\,.
\end{equation}
The main result of \cite{NN4} (see Theorem~2.1 there) states that if
$\bsV\in L^\infty_\text{loc}\big(\Om;\IC^{k\times k}\big)$, and if there exists $\hatt\bsM\geq\bsM$,
$0<\hatt\bsM\in C^\infty\big(\Om;\IC^{k\times k}\big)$, such that $\Om$ endowed with
the Riemannian metric given by
\begin{equation*}
ds^2=\sum_{j,l=1}^d {\hatt\bsM}^{-1}(x)_{jl}\,dx_j\,dx_l
\end{equation*}
is complete, then $\ID$ is essentially self-adjoint.
While in some generic sense this result is optimal (see the discussion in \cite{Fa,NN4}),
it is not an if and only if statement. For example, 
the standard Dirac operator on $\bbR^d\setminus\{0\}$ is not covered by the aforementioned result,
since it does not satisfy the hypothesis, but is essentially self-adjoint in $d\geq2$.

Our work in this paper complements \cite{NN4}, in that we focus on the case where $\bsM$
is not "complete" (in the sense mentioned above), and the essential self-adjointness of
$\ID$ follows from criteria on the potential $\bbV$. 
We set from the very beginning
$\bsE(x)\equiv\mathds 1$, since the general case reduces to this one by a well-known transformation
(see, for example, \cite{J} and \cite{NN4}). 

The question of finding conditions on the potential $\bbV$ ensuring essential self-adjointess on domains 
in $\mathbb R^d$
is by now well understood for second order, Schr\"odinger-type operators (see, for example,
\cite{Br}, \cite{MT}, \cite{NN1}, \cite{NN3}, \cite{PRS}, \cite{RS}, and the references therein).
By contrast, we are not aware of any general results (or any results at all) for
the case of first order, non semi-bounded operators in dimension $d\geq 2$. 
One of the fundamental difficulties in this case is the fact that, unlike for 
Schr\"odinger-type operators, the matrix structure (and not just the size) of the potential function 
$\bbV$ plays a crucial role. This can already be seen from the (elementary) example
of the Dirac operator on $\Om=(a,b)\subset\IR$ with potential $\bbV(x)=v(x)\mathds 1_2$,
which is never essentially self-adjoint, regardless of the behavior of the (scalar) function $v$
near $a$ and/or $b$ (see Section~\ref{S:Ddim1} for details on this example, and for other
potentials for which the operator is essentially self-adjoint).

Our first task is thus to identify classes of potentials for which a certain behavior near the boundary $\partial\Om$ 
ensures essential self-adjointness of the operator. We tackle this in Section~\ref{S:FOSP}, where we
identify such a class of potentials, see Definition~\ref{D:ScaPo}. We call these scalar potentials, 
as they are a generalization of the standard notion of a Lorentz scalar potential for Dirac operators (see \cite{Th}
for more details). Our main general result, Theorem~\ref{T:S}, loosely states that if $\bbV$ 
is a scalar potential such that $\bbV^2$ grows sufficiently fast and its oscillations are not too wild as $x\to\partial\Om$, 
then $\ID_0+\bbV$ is essentially self-adjoint. General potentials are then dealt with perturbatively
in Section~\ref{S:Pert}. Concerning the method of proof we use in Section~\ref{S:FOSP}, 
note first that the proof of the analogous results for Schr\"odinger-type operators (see, for example,
 \cite{Br}, \cite{MT}, \cite{NN1}, \cite{NN3}, \cite{PRS}, \cite{RS}, and the references therein)
uses in an essential way the (lower) semi-boundedness of the operators in question. In particular,
this is true for the method initiated by the first two authors in \cite{NN1}, and which is based on
Agmon-type exponential estimates (see also \cite[Lemma 3.4]{NN3}). So, as it stands, this
method cannot be applied to our current, non semi-bounded, case. 
The way out
is to use Agmon-Combes-Thomas type estimates, in which the semi-boundedness
condition is replaced by the invertibility of a "deformed" operator, $\ID(h)$, (see Section 2) in the spirit of
the Combes-Thomas approach \cite{CT} to exponential decay of eigenfunctions of Schr\"odinger operators. 
For the use of the Combes-Thomas 
approach to exponential decay of eigenfunctions of Dirac operators, see \cite{Hi}, \cite{N}.

Sections~\ref{S:Dscal} through \ref{S:Ddim2} are dedicated to discussing applications of these general results. We 
are especially interested in cases where the class of potentials being considered is not trivial,
and for which $\ID_0$ is not essentially self-adjoint, while $\ID$ is. More precisely, 
we focus on the question of confining Dirac particles to domains $\Om\subset\IR^d$. For simplicity, 
we restrict ourselves to dimensions $d\leq 3$. 
For comments
on more general cases, see Section~\ref{S:Comments}. 

The standard Dirac operator on 
$\Om\subset\IR^d$ is written as $\ID=\ID_{0,d}+\bbV$ on its domain 
$\cD om(\ID)=C_0^1\big(\Om;\IC^k\big)$, with $k=2$ for $d=1,2$ and $k=4$ for $d=3$. 
The first order part $\ID_{0,d}$ encodes the internal structure of Dirac particles, and is given by
\begin{equation}\label{E:D0}
\ID_{0,d}=
\begin{cases}
\sigma_2 D_1 &\quad\text{in } d=1\\
\sigma_1 D_1+\sigma_2 D_2 &\quad\text{in } d=2\\
\alpha_1 D_1+\alpha_2 D_2+\alpha_3 D_3 &\quad\text{in } d=3\,,
\end{cases}
\end{equation}
where $\sigma_j$ are the Pauli and Dirac matrices, respectively, 
and $D_j=-i\frac{\partial}{\partial x_j}$ (see \eqref{E:D.1.2.3}
for more details). The potential $\bbV$ is a Hermitian matrix-valued function which describes
the external forces to which the particles are subjected. At the heuristic level, the evolution
is governed by the partial differential equation
\begin{equation}\label{E:I.6}
i\frac{d}{dt}\bsPsi(\cdot,t)=\ID\bsPsi(\cdot,t)\,.
\end{equation}
At the mathematical level, the essential self-adjointness of $\ID$ is equivalent, via 
Stone's Theorem, to the existence and uniqueness of a unitary evolution in
$L^2\big(\Om;\IC^k\big)$ compatible with \eqref{E:I.6}. The unitarity of the evolution
implies that
\begin{equation*}
\big\|\bsPsi(\cdot,t)\big\|^2=\big\|\bsPsi(\cdot,0)\big\|^2\quad\text{for all } t\in\IR\,,
\end{equation*}
which at the physical level means that, provided the particle is in $\Om$
at time $t=0$, it will remain with probability 1 in $\Om$ for all time. In other words,
the particle is confined to $\Om$ for all time by the external forces represented by the 
potential $\bbV$. 

In Section~\ref{S:Dscal}, we start by answering in the positive the question of 
whether one can achieve confinement for Dirac operators on bounded domains. 
We focus first on Lorentz scalar potentials (aka mass potentials, see \cite[Chapter 4.2]{Th})
which are defined in \eqref{E:D.8}. Motivated by physics literature on quark confinement
(MIT bag model) and on dynamics of electrons in graphene and other nanostructures,
there is recently a large body of work in mathematical physics on Dirac operators
on domains in $\IR^d$; see \cite{ATMR}, \cite{BCTS}, \cite{BFSB}, \cite{BM}, \cite{MOBP}, \cite{OBV}, \cite{SV}, \cite{TOB}, 
and the references therein. However, to the 
best of our knowledge, self-adjointness is obtained in all cases by imposing boundary conditions
which encode additional forces acting on the system. Very recently is was proven
(see \cite{ATMR}, \cite{BCTS}, \cite{SV})
that certain boundary conditions can be obtained via a limiting procedure starting from a 
Dirac operator on the full space $\IR^d$ and sending the value of the Lorentz scalar potential
to infinity outside of $\Om$. In the same spirit, one can also construct Dirac operators
on manifolds embedded in $\IR^d$, see \cite{MOBP}. 

In our work here, we start by noting that Lorentz scalar potentials are scalar according to
our Definition~\ref{D:ScaPo}, and so the general theory results from Section~\ref{S:FOSP} apply.
We therefore obtain several large classes of confining Lorentz scalar potentials. Moreover,
we also discuss the (optimality) question of the lowest possible growth rate for $\bbV^2$
at $\partial\Om$ which guarantees confinement. More precisely, we consider
the case where $\Om\subset\IR^3$ is bounded and convex, $\partial\Om\in C^2$ has co-dimension 1, and
for $\delta(x)=\text{dist}(x,\partial\Om)$ small enough, the Lorentz scalar potential has the form
$\bbV(x)=\lambda\delta(x)^{-1}\beta$. We prove in Theorem~\ref{T:D1S}(ii) that $\ID_0+\bbV$ 
is essentially self-adjoint for $\lambda\geq\frac12$. One can also show, by reducing a rotationally 
symmetric case to a 1-dimensional Dirac operator (see \cite[Chapter 4.6]{Th} and \cite{Wie2}),
that this lower bound on $\lambda$ cannot be improved (see Remark~\ref{R:OD1}).

Section~\ref{S:Ddim1} focuses on the problem of essential self-adjointness of Dirac operators on
intervals in dimension $d=1$. This question is well understood abstractly, which allows us
to use the extension to Dirac operators of the powerful Weyl limit point/limit circle theory, see \cite{Wie2}.
In turn, this leads to precise if and only if criteria for various types of potentials,
as obtained for example in Proposition~\ref{P:M} and Corollary~\ref{C:SMF}. What makes these
types of results particularly interesting is the fact that the essential self-adjointness of Dirac operators 
with rotational symmetry in $d\geq2$ can be reduced to the essential self-adjointness of
1 dimensional Dirac operators via partial wave analysis (see, for example, \cite[Chapter~4.6]{Th} and \cite{Wie2}).

In Section~\ref{S:Ddim2} we turn to the question of confinement criteria for magnetic Dirac operators
$$
\ID_{2,mag}=\sigma_1\big(D_1-\cA_1\big)+\sigma_2\big(D_2-\cA_2\big)
$$
on bounded domains $\Om\subset\bbR^2$. We want to stress that the magnetic potential
$-\sigma_1\cA_1-\sigma_2\cA_2$ is not scalar in the sense of Definition~\ref{D:ScaPo}, 
and so our general results from Sections~\ref{S:FOSP} and \ref{S:Pert} do not apply to this case.
In addition, recall that the physically relevant quantity in this case is not the magnetic potential
$\bscA=\big(\cA_1\,,\,\,\cA_2\big)$, but the associated magnetic field 
$\cB=\frac{\partial\cA_2}{\partial x_1}-\frac{\partial\cA_1}{\partial x_2}$. Our second main result
of this paper is Theorem~\ref{T:M2}, which implies that if
$\cB(x)\geq \frac{1}{2\delta(x)^2}$ as $x\to\partial\Om$, then $\ID_{2,mag}$ is essentially self-adjoint.
The proof of this theorem relies fundamentally on the structure of $\ID_{2,mag}$,
which allows us to use the supersymmetry lemma of \cite{GST}. While this works for general
$\Om$'s in $\IR^2$, the use of the supersymmetric structure of $\ID_{2,mag}$ means that this 
method of proof cannot be extended to higher dimensions. Particular cases with translational and/or
rotational symmetry can sometimes be treated by reducing them to lower dimensional problems
(see Comment 7 in Section~\ref{S:Comments} for an elementary example), but it remains an 
interesting open problem to find a proof for generic situations in 3 (and more) dimensions. Going back to the
result of Theorem~\ref{T:M2}, we prove that the constant $\frac12$ from \eqref{E:TM2.1} is optimal by
using partial wave decomposition for a rotationally symmetric
magnetic field $\cB(x)=\cB(|x|)$ on the
unit disk -- see Proposition~\ref{P:CM} for more details. In the process, we also prove
Proposition~\ref{L:P}, which states that essential self-adjointness for Dirac operators is not affected by
the removal of one point from the physical domain. 

Finally, we close with Section~\ref{S:Comments}, in which we list a number of
additional comments and open questions related to the various topics of the paper.

\section{Confinement by scalar potentials: general theory}\label{S:FOSP}

As outlined in the Introduction, we work in the setting from \cite{NN4}, but with $\bsE\equiv\mathds 1$.
We recall the definition of the velocity matrix
\begin{equation}\label{E:S.1}
\boldsymbol{M_{jk}} = \Tr \big(\bsAj \bsAk\big).
\end{equation}
In \cite{NN4} we showed that if some Riemannian structure associated with $\bsM$ is complete, then $\mathbb{D}$ is 
essentially self-adjoint (irrespective of the behavior of $\mathbb{V}$ as $x \rightarrow \partial \Omega$). We now consider 
the case in which this does not hold true. More precisely, we assume that there exists $M > 0$ 
such that near $\partial \Omega$
\begin{equation}\label{E:S.2}
    \bsM (x) \leq M \delta^m(x), \quad m < 2,
\end{equation}
(notice that $m \geq 2$ implies completeness) where $\delta(x)$ is, as usual, the distance to the boundary of $\Omega$:
\begin{equation}\label{E:S.3}
    \delta(x) := \inf_{y \in \partial \Omega} |x-y|.
\end{equation}
Our problem is to identify classes of potentials $\mathbb{V}$ for which $\mathbb{D}$ is essentially self-adjoint. 
The strategy we follow below is an extension of the method in \cite{NN1,NN3} to operators which are not semibounded, 
in which case  the localization lemma
(see Lemma 3.3 in \cite{NN3}) cannot be used.

Let $h \in C^1(\Omega)$, and define $\mathbb{D}(h)$ as
\begin{equation}\label{E:S.4}
    \mathbb{D}(h) = e^{-h}\mathbb{D}e^h = e^{-h}\mathbb{D}_0 e^h + \mathbb{V}.
\end{equation}
A direct computation gives
\begin{equation}\label{E:S.5}
    \mathbb{D}(h) = \mathbb{D} - i\sigma(\cdot,\nabla h)\,,
\end{equation}
where (see \cite{NN4}) $\sigma(\cdot,\nabla h)$ is the operator of multiplication by the matrix
\begin{equation}\label{E:S.6}
    \sigma(x,\nabla h(x)) = \sum_{j=1}^d \bsAj (x) \frac{\partial h}{\partial x_j} (x).
\end{equation}
Notice that $\sigma(\cdot,\nabla h)$ is well defined on $C_0^1\big(\Omega;\IC^k\big)$.
The main ingredient in proving the essential self-adjointness of $\mathbb{D}$ is the analog of Lemma 3.4 in \cite{NN3}:

\begin{lemma}[The basic inequality]\label{L:LemmaB}
Let $\zeta \in \mathbb R$ and $\bsPsi_\zeta$ a weak solution of $\ID+i\zeta$,
i.e. $\bsPsi_\zeta\in L^2\big(\Om;\IC^k\big)$ with 
\begin{equation*}
\dlab\bsPsi_{\zeta},(\bbD - i\zeta)\bsPhi\drab=0\quad\text{for all } \bsPhi \in C_0^1(\Omega; \C^k)\,.
\end{equation*}
Assume that there exists $C>0$ such that
\begin{equation}\label{E:S.7}
\big\| \big(\mathbb{D}(-h)+i\zeta\big)\bsPhi \big\| 
\geq C \big\| \bsPhi \big\| \quad\text{for all } \bsPhi \in C_0^1(\Omega; \C^k)\,. 
\end{equation}
Then for all $g \in C_0^1(\Omega; [0,1])$,
\begin{equation}\label{E:S.8}
    C \big\| ge^h \bsPsi_{\zeta} \big\| \leq \big\| \sigma (\cdot,\nabla g) e^h \bsPsi_{\zeta} \big\|\,.
\end{equation}
\end{lemma}

\begin{proof}
By assumption, $e^h\bsPhi \in C_0^1\big(\Omega; \C^k\big)$, and so
\begin{equation}\label{E:BI.1}
    0 = \dlab \bsPsi_{\zeta},\big(\bbD- i\zeta\big)e^h \bsPhi\drab 
   =  \dlab e^h\bsPsi_{\zeta},\big(\bbD(h) - i\zeta\big) \bsPhi\drab\,.    
\end{equation}
Let now  $g \in C_0^1\big(\Omega; [0,1]\big)$. For $\bsPhi \in C_0^1\big(\Omega;\C^k\big)$, we note that
$g \bsPhi \in C_0^1\big(\Omega; \C^k\big)$, so using \eqref{E:BI.1} we obtain
\begin{equation}\label{E:BI.2}
    \dlab g e^h  \bsPsi,\big(\bbD(h) - i\zeta\big)\bsPhi  \drab = \dlab e^h \bsPsi, g \big(\bbD(h) -i\zeta \big)\bsPhi \drab 
    =\dlab e^h \bsPsi_{\zeta},[g,\bbD(h)]\bsPhi\drab\,.
\end{equation}
By a direct computation we find that $[g,\bbD(h)] = i\sigma(\cdot,\nabla g)$ on $C_0^1(\Omega; \C^k)$.
Since $g$, and hence also $\nabla g$, has compact support, the operator of multiplication 
with $\sigma(\cdot,\nabla g)$ is bounded and self-adjoint, which allows us to conclude that 
there exists a constant $C'>0$ such that
\begin{equation*}
\big| \dlab ge^h \bsPsi_{\zeta},\big(\bbD(h)-i\zeta\big)\bsPhi  \drab\big| \leq C'\big\| \bsPhi \big\|
 \quad\text{for all }\bsPhi \in C_0^1(\Omega; \C^k)\,.
\end{equation*}
In other words, $ge^h \bsPsi_{\zeta}$ belongs to the domain of the adjoint of $\bbD(h)-i\zeta $:
\begin{equation*}\label{E:BI.6}
    ge^h \bsPsi_{\zeta} \in \cD om\big((\bbD(h)-i\zeta )^*\big)\quad\text{and}\quad
    (\bbD(h)-i\zeta )^\ast ge^h \bsPsi_{\zeta} = -i\sigma(\cdot,\nabla g)e^h \bsPsi_{\zeta}.
\end{equation*}

All of the above show that the conclusion of the lemma, \eqref{E:S.8}, follows if we prove that
\begin{equation}\label{E:BI.7}
 \big\| (\bbD(h)-i\zeta )^* g e^h \bsPsi_{\zeta} \big\| \geq C \big\|g e^h \bsPsi_{\zeta} \big\|\,.
\end{equation}
To achieve this we use Friedrichs \cite{Fr} result concerning the identity of weak and strong extensions of 
first order differential operators
as given in \cite{Ho}.  Let $\bsPhi_{\epsilon}$ be the standard mollification of $ge^h\bsPsi_{\zeta}$:
\begin{equation*}
\bsPhi_{\epsilon}(x)=\int_{\Omega}(ge^h\bsPsi_{\zeta})(x-\epsilon y)\phi(y)dy
\end{equation*}
with $\phi \in C^{\infty}_0(\mathbb R^d)$, $\supp \phi \subset \{x ||x| < 1\}$, $\int_{\mathbb R^d} \phi (x)dx =1$.
Since $ge^h\bsPsi_{\zeta}$ has compact support in $\Omega$, there exists $K \subset \Omega$ compact such that, for 
sufficiently small $\epsilon$, 
\begin{equation*}
\supp \bsPhi_{\epsilon},\,\,\supp (ge^h\bsPsi_{\zeta}) \subset K\,.
\end{equation*}
Note that $(\bbD(h)-i\zeta)^*$ equals the maximal extension of $\bbD(-h)+i\zeta$, see \cite[Proposition 1.14]{Sch}. 
It then follows from \cite[Theorem~3.2]{Ho} that,
for sufficiently small $\epsilon>0$, $\bsPhi_{\epsilon} \in C_0^{\infty}(\Omega; \C^k)$ and
$$ 
\lim _{\epsilon\to 0}  \big\| \bsPhi_{\epsilon}- ge^h\bsPsi_{\zeta} \big\| =0,
$$
\begin{equation}\label{E:BI.8}
\lim_{\epsilon\to0} \big\| \big(\bbD(-h)+i\zeta\big)\bsPhi_{\epsilon}-\big(\bbD(h)-i\zeta\big)^* ge^h\bsPsi_{\zeta} \big\| =0.
\end{equation}
Combining \eqref{E:S.7} and \eqref{E:BI.8} yields \eqref{E:BI.7}, which concludes the proof of the lemma.
\end{proof}

As in the second order PDE case (see \cite{NN3}), in order to conclude the essential self-adjointness of  $\bbD$ 
one has to prove that there exists $a \in [0, \infty)$ for which Lemma \ref{L:LemmaB} implies $\bsPsi_{\pm a}=0$. 
To achieve this, we need to, first,
choose $h$ and a sequence of $g_n$ such that the right hand side of \eqref{E:S.8} becomes negligible 
in the limit $n \rightarrow \infty$, and second,
for a given $h$, find conditions on $\bbV$ for which \eqref{E:S.7} holds.
The choices for $h$ and $g_n$ are essentially the same as those in \cite{NN1} and \cite{NN3} 
(for more refined versions see \cite{NN1}). 

In general, the distance to the boundary, $\delta(x)$, is only Lipschitz 
continuous; see, for example \cite{GT}. To deal with this potential lack of smoothness, we use instead 
a regularized distance to $\partial \Omega$, 
$\hat{\delta}(x)$. The existence of $\hat \delta (x)$ having the properties we need is given by the following  theorem (see 
\cite[Chap. VI Theorem 2]{St})

\begin{theorem}\label{T:delta}
There exists $\hat{\delta}\,:\,\Omega \rightarrow (0,\infty)$, $\hat{\delta} \in C^\infty(\Omega)$, such that
for all $x\in\Om$
\begin{equation*}
\frac{1}{5} \delta(x) \leq \hat{\delta}(x) \leq \frac{4}{3}\,12^d \,\delta(x)
\quad\text{and}\quad
\lvert \nabla \hat{\delta}(x) \rvert \leq A\,12^d \,\sqrt{d} \,,
\end{equation*}
where $A$ is an absolute constant.
\end{theorem}
For later use, we denote the upper bound on the gradient of $\hat\delta$ by 
$L=A\,12^d \,\sqrt{d}\,.$

Let 
\begin{equation}\label{E:S.9}
0<t_0 < \min \left\{\frac{1}{2} \sup_{x \in \Omega} \hat{\delta}(x) , 1 \right\}.    
\end{equation}
We choose $h$ on $\Om$ of the form:
\begin{equation}\label{E:S.10}
    h(x) = G_{t_0,m}(\hat{\delta}(x))
\end{equation}
where $G_{t_0,m}\,:\,(0,\infty) \rightarrow \R$ has the following properties:
\begin{equation}\label{E:S.11}
\begin{aligned}
G'_{t_0,m}(t) \geq 0 \text{ for all } &t >0; \qquad  G'_{t_0,m}(t) = 0 \text{ for all } t \geq 2t_0;\\
&\text{and} \quad    G_{t_0,m}(t) = \left(1-\frac{m}{2}\right)\ln t, \text{ for } t < t_0.
\end{aligned}
\end{equation}

The next step is to choose $g_n$. Let $g_0 \in C^1([0,\infty))$, with $0 \leq g_0 \leq 1$, $0 \leq g'_0 \leq 2$, and
\[
g_0(t) = 
\begin{cases}
0 &\quad \text{for } t \leq 1\,, \\
1 & \quad\text{for } t \geq 2\,.
\end{cases}
\]
Also take $1 > \rho_0 > 0$ sufficiently small and
\begin{equation}\label{E:S.12}
    \rho_n = 2^{-n}\rho_0.
\end{equation}
Then we set
\begin{equation}\label{E:S.13}
    g_n(t) = g_0\left(\frac{t}{\rho_n}\right),
\end{equation}
and finally (with a slight abuse of notation)
\begin{equation}\label{E:S.14}
g_n(x) = g_n(\hat{\delta}(x)).    
\end{equation}
Note that, by construction,
\begin{equation}\label{E:S.15}
    \begin{cases}
    \nabla g_n(x) = 0, & \text{ for } \hat{\delta}(x) \notin (\rho_n, 2\rho_n) \\
    \text{and} & \\
    \lvert \nabla g_n(x) \rvert \leq \frac{2}{\rho_n} \lvert \nabla \hat{\delta}(x) \rvert & \text{ for all } x\in\Om\,.
    \end{cases}
\end{equation}

We now focus on the problem of finding, for $h$ as above, conditions on $\bbV$ for which \eqref{E:S.7} holds.  
As we will show in Lemma~\ref{L:NES}, there are examples of 1-dimensional Dirac operators on finite intervals
$\Omega = (a,b) \in \bbR$ for which $\bbD$ is not essentially self-adjoint irrespective of the  behavior of $\bbV$ near 
$\partial \Omega$. This means that the structure, and not just the size, of the potential $\bbV$ is crucial. 
It turns out that a good class of potentials is described in the following definition:
\begin{definition}\label{D:ScaPo}
A potential $\mathbb{V}$ is called scalar if, for $j = 1,2,\ldots,d$ and for all $x \in \Omega$,
\begin{equation}\label{E:S.16}
    \big\{\bsAj(x),\mathbb{V}(x)\big\} := \bsAj(x)\mathbb{V}(x) + \mathbb{V}(x)\bsAj(x) = 0.
\end{equation}
Alternatively one can write \eqref{E:S.16} as
\begin{equation}\label{E:S.17}
    \big\{\sigma(x,\xi),\mathbb{V}(x)\big\} = 0
\end{equation}
for all $x \in \Omega$, $\xi \in \R^d$.
\end{definition}

Two remarks are in order. The first is that, as we shall see later, the fact that $\bbV$ is scalar is not 
a necessary condition for essential self-adjointness of $\bbD$. Second, in the one-dimensional case, 
potentials of the form $\sigma_2 v_2 + \sigma_3 v_3$ are scalar, in addition to the Lorentz scalar potentials as they are 
defined for three-dimensional standard Dirac operators (see \cite[\S 4.2]{Th}). 

Next, we compute $\big\| (\bbD(-h) + i \zeta)\bsPhi \big\|^2$ for scalar potentials 
$\bbV\in C^1(\Om; \C^{k \times k})$.
\begin{lemma}\label{L:LemmaS}
Let $\mathbb{V}$ be a scalar potential, $\bbV\in C^1(\Om; \C^{k \times k})$. Then the following identity holds for
any $\zeta \in \R$ and for all
$\bsPhi \in C_0^1\big(\Omega;\IC^k\big)$:
\begin{equation}\label{E:S.18}
\begin{aligned}
    \big\| \big(\mathbb{D}(-h) +i\zeta \big)\bsPhi \big\|^2 
    &= \big\|\big(\mathbb{D}_0 + i\sigma(\cdot,\nabla h)\big)\bsPhi \big\|^2 + \dlab  \bsPhi,(\mathbb{V}^2+\zeta^2)\bsPhi \drab \\ 
    &\quad+2\zeta \dlab \bsPhi, \sigma(\cdot, \nabla h) \bsPhi \drab 
    -i \dlab \bsPhi, [\sigma(\cdot,\nabla h),\mathbb{V}]\bsPhi \drab\\
    &\quad- \frac{i}{2} \bdlab \bsPhi, \bigg(\sum_{j=1}^d \bigg[\bsAj,\frac{\partial}{\partial x_j}\mathbb{V} \bigg] \bigg)\bsPhi \bdrab.
\end{aligned}
\end{equation}
\end{lemma}

\begin{proof}
The lemma follows from a straightforward computation using the identity:
    \begin{equation}\label{E:S.19}
        \mathbb{D}_0 \mathbb{V} + \mathbb{V} \mathbb{D}_0 = -\frac{i}{2} \sum_{j=1}^d \left[ \bsAj, \frac{\partial}{\partial x_j} \mathbb{V} \right],
    \end{equation}
valid on $C_0^1\big(\Omega;\IC^k\big)$.
This identity itself can be checked directly since\\ 
$\mathbb{D}_0 = \bsA \cdot \bsD - \frac{i}{2} \diver \bsA$,
and hence
\begin{align*}
        \bbD_0 \bbV + \bbV \bbD_0 &= -i \bsA \cdot \nabla \bbV + \bsA \bbV \cdot \bsD - \frac{i}{2} (\diver \bsA) \bbV - \bsA \bbV \cdot \bsD - \frac{i}{2} \bbV \diver \bsA \\
        &= -i \bsA \cdot \nabla \bbV -\frac{i}{2}(\diver(\bsA \bbV) - \bsA \cdot \nabla \bbV) - \frac{i}{2} (\diver (\bbV \bsA) - \nabla \bbV \cdot \bsA) \\
        &= -\frac{i}{2} \bsA \cdot \nabla \bbV + \frac{i}{2} \nabla \bbV \cdot \bsA,
    \end{align*}
which is exactly \eqref{E:S.19}.
\end{proof}

We can now formulate the main result of this section:

\begin{theorem}\label{T:S}
Let $\mathbb{V}$ be a scalar potential, $\mathbb{V}\in C^1(\Omega; \C^{k \times k})$, and choose
$t_0$ as in \eqref{E:S.9} and $h$ as in \eqref{E:S.10}.
Assume that there exist $M < \infty$, $m < 2$, $\delta_0\in(0,t_0)$, and $c>0$ such that
\begin{equation}\label{E:S.25}
    \quad \bsM (x) \leq M \hat{\delta} (x)^m\quad \text{for all } x\in\Om\,,
\end{equation}
and 
\begin{equation}\label{E:S.26}
\bbV^2 - \frac{i}{2}(\bsA \cdot \nabla \mathbb{V} - \nabla \mathbb{V} \cdot \bsA) -i[\sigma(\cdot,\nabla h),\mathbb{V}]
-\sigma(\cdot,\nabla h)^2 \geq c \mathds{1}
\end{equation}
for all $x \in \Omega$ with $\hat{\delta}(x) < \delta_0$.

Then $\mathbb{D} = \mathbb{D}_0 + \mathbb{V}$ is essentially self-adjoint on $C_0^1(\Omega;\C^{k})$.
\end{theorem}

\begin{proof}
We first verify that \eqref{E:S.7} holds for sufficiently large $|\zeta|$. From \eqref{E:S.18} it is sufficient to show that
\begin{equation}\label{E:S.28}
    \mathbb{V}^2 + \zeta^2  + 2\zeta\, \sigma(\cdot,\nabla h) - i[\sigma(\cdot,\nabla h),\mathbb{V}]  
    - \frac{i}{2}\sum_{j=1}^d \bigg[\bsAj,\frac{\partial}{\partial x_j}\mathbb{V}(x)\bigg] \geq c \mathds{1}\,.
\end{equation}
Now, since $\{x \in \Omega \mid \hat{\delta}(x) \geq \delta_0\}$ is compact, one can choose $|\zeta|$ large enough so that \eqref{E:S.28} holds for all $x$ satisfying $\hat{\delta}(x) \geq \delta_0$. For $\hat{\delta}(x) < \delta_0$, \eqref{E:S.28} follows from \eqref{E:S.26} and the fact that
\[
\zeta^2 + 2\zeta \sigma(x,\nabla h(x)) + \sigma(x,\nabla h(x))^2 \geq 0.
\]

Fix, now, an arbitrary compact $K \subset \Omega$. There exists $\delta_K >0$ such that
\begin{equation}
    K \subset \big\{x \in \Omega\,\big|\, \hat{\delta}(x) >\delta_K\big\}.
\end{equation}
In the definition of $g_n$ we choose $\rho_0 \leq \frac{1}{2}\delta_K$.
Notice that, by construction, $g_n \rvert_K = 1$  for all $n\geq1$.
We now use \eqref{E:S.8} for $g_n$ as defined in \eqref{E:S.12} and \eqref{E:S.13}. On the one hand, since
\begin{equation}
    \inf_{x \in K} e^{2h(x)} \geq \inf_{\hat{\delta}(x) \geq \delta_K} e^{2h(x)} =: H_K^2 > 0,
\end{equation}
one has
\begin{align}\label{E:S.33}
    c^2 \big\| g_n e^h \bsPsi_{\zeta} \big\|^2 \geq c^2 \int_K g_n(x)^2 e^{2h(x)}|\bsPsi_{\zeta}(x)|^2 dx 
    \geq c^2 H_K^2 \int_K |\bsPsi_{\zeta}(x)|^2 dx.
\end{align}
On the other hand, from \eqref{E:S.10}, \eqref{E:S.11}, \eqref{E:S.25}, and the fact that for all $\xi \in \R^d$, $\big\| \sigma(x,\xi)\big\|^2 \leq \langle \xi, \bsM(x) \xi \rangle$ (see \cite[Lemma~2.2]{NN4}), we obtain that
\begin{equation}\label{E:S.34}
    \big\| \sigma(\cdot,\nabla g_n)e^h \bsPsi_{\zeta} \big\|^2 \leq M \int_\Omega \hat{\delta}(x)^2 \lvert \nabla g_n(x) \rvert^2 |\bsPsi_{\zeta}(x)|^2\, dx.
\end{equation}
Let $\Omega_n = \big\{x \in \Omega \mid \hat{\delta}(x) > \rho_n \big\}$. Using \eqref{E:S.12} and \eqref{E:S.15} in \eqref{E:S.34} yields
\begin{equation}\label{E:S.36}
\begin{aligned}
\big\| \sigma(\cdot,\nabla g_n)e^h \bsPsi_{\zeta} \big\|^2 &\leq M \int_{\Omega_{n+1}\setminus \Omega_n} \hat{\delta}(x)^2  
|\nabla g_n(x)|^2 |\bsPsi_{\zeta}(x)|^2 dx \\
&\leq 4L^2M \int_{\Omega_{n+1}\setminus \Omega_n} \big|\bsPsi_{\zeta}(x)\big|^2\, dx.
\end{aligned}
\end{equation}
From \eqref{E:S.33}, \eqref{E:S.36}, and Lemma \ref{L:LemmaB},
\begin{equation}
    \frac{c^2 H_K^2}{4L^2M} \int_K |\bsPsi_{\zeta}(x)|^2 dx \leq \int_{\Omega_{n+1}\setminus \Omega_n} |\bsPsi_{\zeta}(x)|^2 dx.
\end{equation}
Since $\bsPsi_{\zeta} \in L^2(\Omega; \C^{k})$, the right-hand side of \eqref{E:S.36} converges to zero 
as $n \rightarrow \infty$. Thus one obtains that $\chi_K \bsPsi_{\zeta} = 0$. Since $K$ is arbitrary, 
this implies that for $a>0$ sufficiently large,  $\bsPsi_{\pm ia} = 0$. Hence $\mathbb{D}$ has defect indices equal to zero, 
proving the claim that $\mathbb{D}$ is essentially self-adjoint.
\end{proof}

\begin{remark}
It is not a-priori obvious that, for a given $\ID$, there exist scalar potentials as described in
Definition~\ref{D:ScaPo}. For example, if $k=1$, or $\bsAj(x) = Q(x)\mathds{1}$, $j=1,\ldots,d$, 
then nontrivial scalar potentials do not exist, and thus Theorem~\ref{T:S} is void. However,
as already mentioned in the Introduction, scalar potentials do exists for standard Dirac operators,
and applications of Theorem~\ref{T:S} will be discussed at length in Section~\ref{S:Dscal}.
\end{remark}

In the proof of Theorem \ref{T:S} we used only the fact that
        \begin{equation}\label{E:S.38}
            \big\| (\bbD_0 + i \sigma(\cdot, \nabla h))\bsPhi \big\| \ge 0.
        \end{equation}
However, in certain cases one can prove Hardy  type inequalities:
 \begin{equation}\label{E:S.39}
 \big\| (\bbD_0 + i\sigma(\cdot,\nabla h))\bsPhi \big\|^2 \ge \int_\Omega H_h(x) \big| \bsPhi(x) 
\big|^2\,dx,\quad\text{for all } \bsPhi \in C_0^1(\Omega;\C^k)\,,
 \end{equation}
where $H_h(x)$ is uniformly bounded from below and blows up as $x \rightarrow \partial \Omega$.
            
 In such a case, Theorem \ref{T:S} takes the form
\begin{theorem}\label{T:SH}
Let $\mathbb{V}$ be a scalar potential, $\mathbb{V}\in C^1(\Omega; \C^{k \times k})$, and choose
$t_0$ as in \eqref{E:S.9} and $h$ as in \eqref{E:S.10}.
Assume that there exists a Hardy function $H_h$ as defined in \eqref{E:S.39}, and constants
$M < \infty$, $0 \leq m < 2$, $\delta_0\in(0,t_0)$, and $c>0$ such that \eqref{E:S.25} holds and
\begin{equation}\label{E:S.40}
H_h\mathds 1+\bbV^2 - \frac{i}{2}(\bsA \cdot \nabla \mathbb{V} - \nabla \mathbb{V} \cdot \bsA) -i[\sigma(\cdot,\nabla h),\mathbb{V}]
-\sigma(\cdot,\nabla h)^2 \geq c \mathds{1}
\end{equation}
for all $x \in \Omega$ with $\hat{\delta}(x) < \delta_0$.

Then $\mathbb{D} = \mathbb{D}_0 + \mathbb{V}$ is essentially self-adjoint on $C_0^1(\Omega;\C^{k\times k})$.
\end{theorem}

\section{Confinement by general potentials: a perturbative result}\label{S:Pert}

In the previous section we obtained sufficient conditions on a scalar potential $\bbV_s$ which ensure
essential self-adjointness of $\bbD_0+\bbV_s$ with domain $C_0^1\big(\Omega;\IC^k\big)$. In this section,
we consider the same question for an operator of the form $\bbD_0+\bbV_s+\bbW$ in a 
perturbative regime, i.e. when $\bbD_0+\bbV_s$ is essentially self-adjoint and $\bbW$ is
a general potential which is small enough. Our result is a consequence of well-known
Kato-Rellich types theorems, and it is given in the following theorem.

\begin{theorem}\label{T:P}
Let $\bbV_s\in C^1\big(\Omega;\IC^{k\times k}\big)$ be a scalar potential and 
$\bbW\in C^0\big(\Omega;\IC^{k\times k}\big)$ a general one.
Assume that $\bbD_0+\bbV_s$ is essentially self-adjoint on $C_0^1\big(\Omega;\IC^k\big)$. Further assume that there exists a function $H_0$ bounded from below on $\Om$ such that
$H_0(x)\to\infty$ as $x\to\partial\Om$ and 
\begin{equation}\label{E:S.39a}
                    \big\| \bbD_0  \bsPhi \big\|^2 \ge \int_\Omega H_0(x) \lvert \bsPhi(x) 
                    \rvert^2\,dx\,,\quad\text{for all }\bsPhi\in C_0^1\big(\Omega;\IC^k\big) \,.
\end{equation}
\begin{enumerate}
\item[i.] If there exists $\delta_0>0$ such that
\begin{equation}\label{E:P.1}
\chi_{_{\delta_0}}\bigg(H_0\mathds 1+\bbV_s^2
-\frac{i}{2}\sum_{j=1}^d \left[\bsAj,\frac{\partial\bbV_s}{\partial x_j}\right]-\bbW^2\bigg)\geq0
\end{equation}
where $\chi_{_{\delta_0}}$ denotes the characteristic function of the set 
$\big\{x\in\Omega\,\big|\,\delta(x)\leq\delta_0\big\}$,
then $\bbD_0+\bbV_s+\bbW$ is essentially self-adjoint on $C_0^1\big(\Omega;\IC^k\big)$.
\item[ii.] If there exist $0<C<1$ and $\delta_0>0$ such that
\begin{equation}\label{E:P.2}
\chi_{_{\delta_0}}\left(C\bigg(H_0\mathds 1+\bbV_s^2
-\frac{i}{2}\sum_{j=1}^d \left[\bsAj,\frac{\partial\bbV_s}{\partial x_j}\right]\bigg)-\bbW^2\right)\geq0\,,
\end{equation}
then $\overline{\bbD_0+\bbV_s+\bbW}$ is self-adjoint on $\cD\big(\overline{\bbD_0+\bbV_s}\big)$,
where $\overline{\bbD_0+\bbV_s+\bbW}$ and $\overline{\bbD_0+\bbV_s}$ denote
the closures of $\bbD_0+\bbV_s+\bbW$ and $\bbD_0+\bbV_s$, respectively.
\end{enumerate}
\end{theorem}

\begin{proof}
Set $\bbZ_s=\bbV_s^2-\frac{i}{2}\sum_{j=1}^d \left[\bsAj,\frac{\partial\bbV_s}{\partial x_j}\right]$.
Note that hypotheses \eqref{E:P.1} and \eqref{E:P.2} can be rewritten as
\begin{equation}\label{E:P.5}
\dlab\bsPhi,\chi_{_{\delta_0}}\bbW^2\bsPhi\drab\leq 
C\dlab\bsPhi,\chi_{_{\delta_0}}\big(H_0\mathds 1+\bbZ_s\big)\bsPhi\drab
\end{equation}
with $C=1$ and $C\in\big(0,1\big)$, respectively. 
Since $\bbW^2,\bbZ_s\in C^0\big(\Omega;\IC^{k\times k}\big)$ and $H_0$
is uniformly bounded from below, there exists a constant $b_{_{\delta_0}}<\infty$ such that
\begin{equation}\label{E:P.7}
\sup_{\Omega} \big\{(1-\chi_{_{\delta_0}})|\bbW^2|,(1-\chi_{_{\delta_0}})|\bbZ_s|,(1-\chi_{_{\delta_0}})|H_0|\big\}
\leq b_{_{\delta_0}}
\end{equation}
This bound, together with \eqref{E:P.5} and the decomposition $1=\chi_{_{\delta_0}}+(1-\chi_{_{\delta_0}})$, implies that
\begin{equation}\label{E:P.8}
\dlab\bsPhi,\bbW^2\bsPhi\drab\leq C\dlab\bsPhi,\big(H_0\mathds 1+\bbZ_s\big)\bsPhi\drab
+3b_{_{\delta_0}}\|\bsPhi\|^2\,.
\end{equation}
The Hardy barrier assumption \eqref{E:S.39a} and Lemma~\ref{L:LemmaS} imply that
on $C_0^1\big(\Omega;\IC^k\big)$:
\begin{equation}\label{E:P.3}
\big\|(\bbD_0+\bbV_s)\bsPhi\big\|^2
=\|\bbD_0\bsPhi\|^2+\dlab\bsPhi,\bbZ_s\bsPhi\drab
\geq \dlab\bsPhi,\big(H_0\mathds1+\bbZ_s\big)\bsPhi\drab\,,
\end{equation}
which combines with \eqref{E:P.8} to yield
\begin{equation}\label{E:P.9}
\big\|\bbW\bsPhi\big\|\leq \sqrt{C}\big\|(\bbD_0+\bbV_s)\bsPhi\big\|+\sqrt{3b_{_{\delta_0}}} \big\|\bsPhi\big\|\,.
\end{equation}
Here, as above, $C=1$ in part i. of the theorem, and $C\in\big(0,1\big)$ in part ii.

Now let $\overline{\bbW}$ and $\overline{\bbD_0+\bbV_s}$ denote the closures of 
$\bbW$ and $\bbD_0+\bbV_s$ with $\cD om(\bbW)=\cD om(\bbD_0+\bbV_s)=C_0^1\big(\Omega;\IC^k\big)$,
respectively. $\overline{\bbW}$ is symmetric and, by assumption, $\overline{\bbD_0+\bbV_s}$ is
self-adjoint. In addition, standard limiting arguments using \eqref{E:P.9} imply that
\begin{equation}\label{E:P.10}
\cD om\big(\overline{\bbD_0+\bbV_s}\big)\subset \cD om\big(\overline{\bbW}\big)
\end{equation}
and
\begin{equation}\label{E:P.11}
\big\|\overline{\bbW}\bsPsi\big\|\leq \sqrt{C}\,\big\|\overline{\bbD_0+\bbV_s}\,\bsPsi\big\|+\sqrt{3b_{_{\delta_0}}} \big\|\bsPsi\big\|
\quad\text{for all }\bsPsi\in \cD om\big(\overline{\bbD_0+\bbV_s}\big)\,.
\end{equation}
Since we know, by assumption, that $C_0^1\big(\Om;\IC^k\big)$ is a core of $\overline{\bbD_0+\bbV_s}$, 
the statement of i. follows directly from \eqref{E:P.11} with $C=1$ and W\"ust's Theorem \cite[Theorem X.14]{RS}.
Part ii. of Theorem~\ref{T:P} similarly follows from \eqref{E:P.11} with $C<1$ and the standard Kato-Rellich Theorem
\cite[Theorem X.12]{RS}.
\end{proof}

\section{Confinement of Dirac particles: Lorentz scalar potentials in $d\leq 3$}\label{S:Dscal}

Let $\Omega$ be a (bounded) domain in $\IR^d$, with $d\leq 3$, on which
we consider the standard Dirac operator $\bbD$ as follows:
\begin{equation}\label{E:D.1.2.3}
\bbD=
\begin{cases}
\sigma_2D_1+\bbV\,,&\quad \cD om(\bbD)=C_0^1\big(\Om;\IC^2\big) \text{ for } d=1\\
\sigma_1D_1+\sigma_2D_2+\bbV\,,&\quad \cD om(\bbD)=C_0^1\big(\Om;\IC^2\big) \text{ for } d=2\\
\alpha_1D_1+\alpha_2 D_2+\alpha_3D_3+\bbV
\,,&\quad\cD om(\bbD)=C_0^1\big(\Om;\IC^4\big) \text{ for } d=3
\end{cases}
\end{equation}
Here and in what follows we use the standard notation
\begin{equation}\label{E:D.4}
\sigma_1=\left(\begin{matrix}0 & 1\\ 1 & 0\end{matrix}\right)\,,\quad
\sigma_2=\left(\begin{matrix}0 & -i\\ i & 0\end{matrix}\right)\,,\quad
\sigma_3=\left(\begin{matrix}1 & 0\\ 0 & -1\end{matrix}\right)\,,\quad
\end{equation}
for the Pauli $2\times2$ matrices, and
\begin{equation}\label{E:D.5}
\alpha_j=\left(\begin{matrix}0 & \sigma_j\\ \sigma_j & 0\end{matrix}\right)\,,\,\,\, j=1,2,3,\qquad
\beta=\left(\begin{matrix}\mathds 1_2 & 0\\0&-\mathds1_2\end{matrix}\right)
\end{equation}
for the Dirac $4\times4$ matrices. In addition, we consider potentials $\bbV$ which are
Hermitian matrix-valued functions on $\Omega$, $2\times2$ for $d=1,2$ and $4\times4$ for $d=3$,
with $C^1(\Omega)$-smooth matrix entries.

Note that for any $d\leq3$, the Dirac operators are of the form $\bbD=\bbD_{0}+\bbV$, where $\ID_{0,d}$
are as in \eqref{E:D0} and \eqref{E:D.1.2.3}.
Other forms of the operators $\bbD_{0}$, unitarily equivalent to the ones above, 
sometimes occur; see, e.g., \cite[Appendix to Chapter 1]{Th}.
One simple such example in the $d=1$ case is given by $\Sigma D_1$ for any $2\times2$ Hermitian
matrix $\Sigma$ with $\Sigma^2=\mathds 1_2$ and Tr$(\Sigma)=0$. 
Since any two such matrices $\Sigma$ have the same nondegenerate
eigenvalues (namely $\pm1$), any two such representations $\Sigma D_1$ are unitarily equivalent.

We collect here, for later use, the most important (anticommutation) rules for Pauli and Dirac matrices:
\begin{equation}\label{E:D.7}
\{\sigma_j,\sigma_\ell\}=2\delta_{j\ell}\mathds 1_2\,,\quad \{\alpha_j,\alpha_\ell\}=2\delta_{j\ell}\mathds 1_4\,,\quad
\{\alpha_j,\beta\}=0\,,\quad \text{and} \quad \beta^2=\mathds 1_4
\end{equation}
for any $j,\ell=1,2,3$, where for any two $k\times k$ matrices $\gamma_1,\gamma_2$, the anticommutator is
defined as usual by 
$$
\{\gamma_1,\gamma_2\}=\gamma_1\gamma_2+\gamma_2\gamma_1\,.
$$

In this section we consider a class of scalar potentials $\bbV_{Ls}\in\cS_{Ls}$ given as
\begin{equation}\label{E:D.8}
\bbV_{Ls}(x)=
\begin{cases}
\sigma_1 v(x)\,, &\quad\text{for }d=1\\
\sigma_3 v(x)\,, &\quad\text{for }d=2\\
\beta v(x)\,, &\quad\text{for }d=3\\
\end{cases}
\end{equation}
where $v\in C^1(\Omega,\IR)$. Our goal in this section is to find classes of real-valued
scalar functions $v$ for which $\bbD_{0,d}+\bbV_{Ls}$ is essentially self-adjoint. We consider
the class of potentials $\cS_{Ls}$ since, on the one hand, the results and proofs in this case are 
simple enough to clearly illustrate the main ideas, and, on the other hand, it contains for
$d=3$ the Lorentz scalar potentials, which are defined in quantum mechanics via their 
behavior under Lorentz transformations (see \cite[Chap. 4.2]{Th}).

The results of this section follow from applications of Theorems~\ref{T:S} and \ref{T:P}. Since,
by a short calculation using the anticommutation relations above,
the velocity matrix $\bsM$ in \eqref{E:S.1} for Dirac operators is a multiple of the identity matrix,
the choice of exponent in hypothesis \eqref{E:S.25} is $m=0$.
Thus we must provide criteria which guarantee that hypothesis \eqref{E:S.26} holds with $m=0$, i.e. with
\begin{equation}\label{E:D.9}
h(x)=\ln\,\hat\delta(x)\,.
\end{equation}
Recall that $\hat\delta(x)$ is defined via Theorem~\ref{T:delta}, and this choice of
$h$ is consistent with \eqref{E:S.10} and \eqref{E:S.11}.
The two propositions below are direct corollaries of Theorem~\ref{T:S} and
provide large classes of potentials $\bbV_{Ls}$ ensuring the essential self-adjointness
of $\bbD_{0,d}+\bbV_{Ls}$.  
Before we proceed, note that in what follows we will say that a property hold for
a sufficiently small $t$ as a shorthand for saying that there exists $t_0>0$ such that 
the property holds for all $t\in(0,t_0]$.

For $\alpha>1$, we denote by $\cV^\alpha$ the class of scalar-valued,
real-differentiable functions $v$ on $\Omega$ for which there exists $\varepsilon>0$
and a constant $C_v\in(0,\infty)$ such that
\begin{equation}\label{E:D.10}
\frac{C_v}{\delta(x)^\alpha}\leq\big|v(x)\big|\leq \frac{C_v}{\delta(x)^{2\alpha-1-\varepsilon}} 
\quad\text{and}\quad
\big|\nabla v(x)\big|\leq \frac{C_v}{\delta(x)^{2\alpha-\varepsilon}} 
\end{equation}
whenever $\delta(x)$ is sufficiently small. Loosely speaking, condition
\eqref{E:D.10} says that a function $v\in\cV^\alpha$ must blow-up fast enough
as $x\to\partial\Omega$, and its partial derivatives cannot behave too wildly.
For example, $v(x)$ proportional to $1/\hat\delta(x)^\alpha$, where $\hat\delta$
is the regularized distance from Theorem~\ref{T:delta}, is in $\cV^\alpha$ since it
satisfies \eqref{E:D.10} for any $0<\varepsilon\leq\alpha-1$.

\begin{proposition}\label{P:SDalpha}
Let $\alpha>1$. For any $d\leq 3$, the Dirac operator
$\ID=\ID_{0,d}+\bbV_{Ls}$ is essentially self-adjoint,
provided its Lorentz scalar potential $\bbV_{Ls}$
is such that $v\in\cV^\alpha$.
\end{proposition}

\begin{proof}
The proof is essentially the same for all $d\leq 3$, so we will only give it for the case $d=3$.
To apply Theorem~\ref{T:S} and conclude essential self-adjointness, 
we need to check that for $\bbV_{Ls}=\beta v$, with $v\in\cV^\alpha$ for some $\alpha>1$,
hypothesis \eqref{E:S.26} holds. Plugging this form of $\bbV_{Ls}$ into \eqref{E:S.26}, with
$h$ as in \eqref{E:D.9}, shows that it is sufficient to prove that there exists a constant $c>0$
such that, for $x$ with $\hat\delta(x)$ sufficiently small,
\begin{equation}\label{E:D.11}
v^2\mathds1_4-\beta\sum_{j=1}^3\alpha_j D_jv-
2\beta\frac{v}{\hat\delta}\sum_{j=1}^3\alpha_jD_j\hat\delta-
\frac{1}{\hat\delta^2}\big|\nabla\hat\delta\big|^2\mathds1_4\geq c\mathds1_4\,.
\end{equation}
From \eqref{E:D.10} and Theorem~\ref{T:delta}, we can estimate, for $\delta(x)$ small enough,
the size of each term above, and since $\alpha>1$, we see that the first term  
on the left-hand side of \eqref{E:D.11} is dominant, ensuring that 
the inequality \eqref{E:D.11} holds for $\delta(x)$ small enough, completing the proof.
\end{proof}

Now consider $\alpha<1$, and a function $v(x)$ proportional to $\frac{1}{\hat\delta(x)^\alpha}$.
In this case, one sees that the fourth term in \eqref{E:D.11}, $|\nabla\hat\delta(x)|^2/\hat\delta(x)^2$,
is (generically) dominant as $\delta(x)\to0$, and so the inequality \eqref{E:D.11} cannot hold true. In fact, in the next section
we give examples in $d=1$ where for such a behavior of $v$ we know that $\bbD_1$ is not essentially
self-adjoint.

So the only case still pending is $\alpha=1$. Let $\cV^1$ be the class of real-valued, differentiable
functions $v$ on $\Omega$ 
for which there exists a constant $C_v\in(0,\infty)$ 
such that
\begin{equation}\label{E:D.12}
v(x)=\frac{\ell(x)}{\delta(x)} 
\quad\text{and}\quad
\big|\nabla v(x)\big|\leq \frac{C_v\big|\ell(x)\big|}{\delta(x)^{2}}\,,\quad\text{with }\big|\ell(x)\big|\geq 1
\end{equation}
whenever $\delta(x)$ is sufficiently small.

\begin{proposition}\label{P:DS1}
Let $v\in \cV^1$. Then there exists $\lambda_v>0$ such that, for any $\lambda\geq\lambda_v$
and for $\tilde v(x)=\lambda\,v(x)$, the Dirac operator $\ID=\ID_{0,d}+\bbV_{Ls}$, whose 
Lorentz scalar potential $\bbV_{Ls}$ is defined using $\tilde v$, is essentially self-adjoint.
\end{proposition}

\begin{proof}
Again, we give the proof only for $d=3$.
Fix $v\in \cV^1$ and let $\tilde v(x)=\lambda v(x)$. As in the previous proof, we use \eqref{E:D.12} to estimate the size, 
for $\delta(x)$ sufficiently small, of all the terms on the lhs of \eqref{E:D.11}. Since all four terms are of the same order,
namely $1/\delta(x)^2$, condition \eqref{E:D.11} is satisfied if, for certain, fixed constants $C_1, C_2, C_3>0$,
the inequality
\begin{equation}\label{E:D.14}
C_1 \lambda^2\ell(x)^2-C_2\lambda\big|\ell(x)\big|-C_3>0
\end{equation}
holds for all $x\in\Omega$ with $\delta(x)$ small enough. Since in this region of $\Omega$ we know 
$\big|\ell(x)\big|\geq 1$, we can conclude that \eqref{E:D.14} holds true
for $\lambda$ large enough, uniformly in $\big|\ell(x)\big|\geq1$.
\end{proof}

The result of Proposition~\ref{P:DS1} can be refined further if we make
more specific assumptions about the domain $\Omega$ and the function $v(x)$.
First, assume that the boundary of $\Omega$  is a 
codimension 1, $C^2$-smooth manifold in $\IR^d$. Then we know, see
e.g. \cite{GT} or \cite[Lemma~6.2]{Br}, that, for $\delta(x)$ sufficiently small,
$\delta$ is $C^1$-smooth and $\big|\nabla\delta(x)\big|=1$. Furthermore,
by a weak Hardy inequality which holds for $\bsPhi\in C^1_0(\Omega; \C^k)$, we
know that (see, e.g., \cite{Dav})
\begin{equation}\label{E:D.15}
\big\|\bbD_{0}\bsPhi\big\|^2=\big\|\nabla\bsPhi\big\|^2\geq 
\dlab\bsPhi,\big(\tfrac{1}{4\delta^2}-h_0(\Omega)\big)\bsPhi\drab\,,
\end{equation}
with $h_0(\Omega)<\infty$. If, in addition, $\Omega$ is convex, then
\begin{equation}\label{E:D.16}
h_0(\Omega)<0\,.
\end{equation}

Let $\mu\in[0,\infty)$. We now consider a subfamily $\cV^1_\mu\subset \cV^1$, consisting
of all functions $v=\frac{\ell}{\delta}\in \cV^1$ such that
\begin{equation}\label{E:D.17}
\limsup_{\delta(x)\to0} \frac{\big|\nabla\ell(x)\big|\delta(x)}{\big|\ell(x)\big|}=\mu\,.
\end{equation}
For clarity and readability, we state and prove the next theorem for the case $d=3$.
The corresponding statements for the cases $d=1$ and $d=2$ are similar, and are left to the interested 
reader.

\begin{theorem}\label{T:D1S}
Assume that $\Omega$ 
has a $C^2$-smooth,
codimension 1 boundary, $\partial\Omega$, in $\IR^3$.
\begin{itemize}
\item[i.] Let $\mu\in[0,\infty)$ and $v\in\cV^1_\mu$. Then for all $\lambda>\frac{1+\mu}{2}$, the Dirac operator
\begin{equation*}
\bbD=\sum_{j=1}^3 \alpha_j D_j+\lambda\beta v\,,\quad\cD(\bbD)=C^1_0\big(\Omega;\IC^4\big)
\end{equation*}
is essentially self-adjoint and the domain of its self-adjoint extension, $\cD(\overline{\bbD})$,
is independent of $\lambda$.
\item[ii.] Assume, in addition, that $\Omega$ is convex, and that, for sufficiently small $\delta(x)$,
$v(x)=\frac{1}{\delta(x)}$. Then the Dirac operator $\bbD$ defined above is essentially self-adjoint
for all $\lambda\geq\frac12$. 
\end{itemize}
\end{theorem}

\begin{proof} i.  Given $\mu\geq0$, choose $v\in \cV^1_\mu$ and $\lambda>\frac{1+\mu}{2}$. Recall that, by definition,
this means that there exist constants $\delta_v,C_v\in(0,\infty)$ and a function $\ell$ such that
\begin{equation*}
v(x)=\frac{\ell(x)}{\delta(x)}\,,\quad \big|\nabla v(x)\big|\leq \frac{C_v\big|\ell(x)\big|}{\delta(x)^{2}}\,,\quad\big|\ell(x)\big|\geq 1
\quad\text{for }\delta(x)<\delta_v\,,
\end{equation*}
and
\begin{equation*}
\limsup_{\delta(x)\to0} \frac{\big|\nabla\ell(x)\big|\delta(x)}{\big|\ell(x)\big|}=\mu\,.
\end{equation*}
Since $\cV^1_\mu\subset\cV^1$, Proposition~\ref{P:DS1} applies to $v$, and hence there exists $\lambda_v>0$ such that
the operator $\bbD_{0}+\tilde\lambda \beta v$ is essentially self adjoint on its
domain,\\ $\cD om(\bbD_{0}+\tilde\lambda \beta v)=C_0^1\big(\Omega;\IC^4\big)$,
for any $\tilde\lambda\geq\lambda_v$.

We will prove the statement in i. by showing that
there exist constants $\delta_0>0$, $a\geq \lambda_v$ and $C\in \big[\frac12,1\big)$ such that
the realization of condition \eqref{E:P.2} in this context holds, namely
\begin{equation*}
\chi_{_{\delta_0}}\left(C\bigg(H_0\mathds 1+\bbV_s^2
-\frac{i}{2}\sum_{j=1}^3 \left[\bsAj,\frac{\partial\bbV_s}{\partial x_j}\right]\bigg)-\bbW^2\right)\geq0
\end{equation*}
with
\begin{equation}\label{E:Ir.1}
\bbV_s=(\lambda+a)\beta v\quad\text{and}\quad \bbW=-a\beta v\,.
\end{equation}
If this is the case, and keeping in mind that $\lambda+a>\lambda_v$ means that
Proposition~\ref{P:DS1} applies to $\bbV_s$ defined here, then Theorem~\ref{T:P}.ii. yields the desired conclusion
for the operator
\begin{equation}\label{E:D.21}
\bbD_{0}+\bbV_s+\bbW=\sum_{j=1}^3 \alpha_j D_j+\tilde\lambda\beta v=\ID\,.
\end{equation}
Specializing \eqref{E:P.2} to our Dirac operators, we recall that $\bsAj=\alpha_j$ for each $1\leq j\leq3$, 
and so the anticommutation relations \eqref{E:D.7} together with the structure of $v=\ell/\delta$ imply that
\begin{equation}\label{E:Ir.2}
\left[\bsAj,\frac{\partial\bbV_s}{\partial x_j}\right]=
-2\frac{\lambda+a}{\delta^2}\beta\alpha_j\,\left(\frac{\partial\ell}{\partial x_j}\,\delta-\ell\,\frac{\partial\delta}{\partial x_j}\right)\,.
\end{equation}
Plugging in the form of $H_0$ from \eqref{E:D.15} and the forms of the potentials in \eqref{E:Ir.1} we see
that we must prove the following:

\textit{Claim 1. Given $\mu\geq 0$ and $\lambda>\frac{1+\mu}{2}$, there exist constants $\delta_0\in(0,\delta_v)$, $a\geq \lambda_v$, and $C\in \big[\frac12,1\big)$ such that
\begin{equation}\label{E:Ir.3}
C\left(\left[\frac14-h_0\big(\Omega\big)\delta^2
+(\lambda+a)^2\ell^2\right]\mathds1_4+(\lambda+a)\ell B\right)-a^2\ell^2\mathds1_4\geq 0
\end{equation}
for all $x\in\Omega$ with $\delta(x)<\delta_0$, where
\begin{equation}\label{E:Ir.4}
B(x)=i\sum_{j=1}^3 \beta\alpha_j\left(-\frac{\partial\delta}{\partial x_j}(x)+
\frac{\frac{\partial\ell}{\partial x_j}(x)\,\delta(x)}{\ell(x)}\right)\,.
\end{equation}}

Using again the anticommutation relations, we immediately see that
the matrix $B(x)$ is Hermitian and
\begin{equation*}
B(x)^2 =\sum_{j=1}^3  \left(-\frac{\partial\delta}{\partial x_j}(x)+
\frac{\frac{\partial\ell}{\partial x_j}(x)\,\delta(x)}{\ell(x)}\right)^2 \mathds1_4\leq \left( \big|\nabla\delta(x)\big|
+\frac{\big|\nabla\ell(x)\big|\delta(x)}{\big|\ell(x)\big|} \right)^2 \mathds1_4\,.
\end{equation*}
Recalling that $|\nabla\delta(x)|=1$, a straightforward diagonalization argument then shows that
\begin{equation}\label{E:Ir.5}
\ell(x)B(x)\geq -\big|\ell(x)\big|\left( 1
+\frac{\big|\nabla\ell(x)\big|\delta(x)}{\big|\ell(x)\big|} \right) \mathds1_4\,.
\end{equation}
Now let $\eta \in (0, 2\lambda_v)$. By hypothesis, there exists $\delta_{\eta}\in(0,\delta_v)$ such that
\begin{equation}\label{E:Ir.6}
h_0(\Omega)\delta_{\eta}^2<\eta \quad\text{and}\quad 
\frac{\big|\nabla\ell(x)\big|\delta(x)}{\big|\ell(x)\big|}<\mu+\eta
\,\,\,\text{ for } \delta(x)<\delta_{\eta}\,.
\end{equation}
Combining all of these inequalities, we conclude that, for $\delta(x)<\delta_0\leq \delta_{\eta}$, $a\geq\lambda_v$,
and $C\in\big[\frac12,1\big)$, the following lower bound holds:
\begin{equation*}
\begin{aligned}
\text{lhs of \eqref{E:Ir.3}}\geq
C\left(\frac14-\eta+(\lambda+a)^2\ell(x)^2-(\lambda+a)(1+\mu+\eta)\big|\ell(x)\big|\right)\mathds1_4
-a^2\ell(x)^2\mathds1_4\,.
\end{aligned}
\end{equation*}
This shows that \text{Claim 1} follows if we prove the following:

\textit{Claim 2. Given $\mu\geq0$ and $\lambda>\frac{1+\mu}{2}$, there exist constants 
$\\eta\in(0,2\lambda_v)$,
$a\geq\lambda_v$, and $C\in \big[\frac12,1\big)$ such that
\begin{equation}\label{E:Ir.7}
\big[C(\lambda+a)^2-a^2\big]y^2-C(\lambda+a)(1+\mu+\eta)y+C\big(\tfrac14-\eta\big)\geq 0
\quad\text{for all }y\geq1\,.
\end{equation}}
Let $F\,:\, \bbR\to\bbR$ be the function on the left-hand side of \eqref{E:Ir.7}
\begin{equation*}
F(y)=\big[C(\lambda+a)^2-a^2\big]y^2-C(\lambda+a)(1+\mu+\eta)y+C\big(\tfrac14-\eta\big)\,.
\end{equation*}
Since $F$ is a quadratic polynomial, \eqref{E:Ir.7} is implied by the following three conditions: 
The coefficient of the quadratic term is strictly positive, i.e.
\begin{equation}\label{E:Ir.8}
C(\lambda+a)^2-a^2>0\,,
\end{equation}
the value of $y$ where $F$ attains its minimum is less than or equal to 1, i.e.
\begin{equation}\label{E:Ir.9}
2\big[C(\lambda+a)^2-a^2\big] - C(\lambda+a)(1+\mu+\eta)\geq 0\,,
\end{equation}
and $F(1) \geq 0$, i.e.
\begin{equation}\label{E:Ir.10}
\big[C(\lambda+a)^2-a^2\big]-C(\lambda+a)(1+\mu+\eta)+C\big(\tfrac14-\eta\big) \geq 0\,.
\end{equation}

The first remark is that, since $C(\lambda+a)(1+\mu+\eta) >0$, \eqref{E:Ir.8} is implied by \eqref{E:Ir.9}.
Consider now \eqref{E:Ir.9}. We regard the left-hand side as a 
quadratic polynomial in $\lambda+a$, with positive dominant coefficient $2C\geq 1$ and
positive discriminant $C^2(1+\mu+\varepsilon_1)^2+16Ca^2>0$. Thus the quadratic expression
is non-negative whenever the variable is above the right-most root. That is, \eqref{E:Ir.9} holds true if
\begin{equation}\label{E:Ir.11}
\lambda+a\geq \frac{1+\mu+\eta+\sqrt{(1+\mu+\eta)^2+\frac{16a^2}{C}}}{4}\,.
\end{equation}
Since
\begin{equation*}
\sqrt{(1+\mu+\eta)^2+\frac{16a^2}{C}}\leq 1+\mu+\eta+\frac{4a}{\sqrt C}\,,
\end{equation*}
we conclude that  \eqref{E:Ir.11} (hence  \eqref{E:Ir.9}) is implied by
\begin{equation}\label{E:Ir.12}
\lambda\geq \frac{1+\mu}{2}+\frac{\eta}{2}+\frac{1-\sqrt C}{\sqrt C}\,a\,.
\end{equation}
Finally, we consider \eqref{E:Ir.10}. By completing a square, \eqref{E:Ir.10} rewrites as 
\begin{equation}\label{E:Ir.13}
C\left(\lambda+a-\frac{1+\mu+\eta}{2}\right)^2-a^2\geq C\, \frac{\big(1+\mu+\eta\big)^2}{4}
-C\left(\frac14-\eta\right)\,.
\end{equation}
Since by assumption
$
a \geq \lambda _v  > \eta/2 ,
$
one has
\begin{equation*}
\sqrt{C}\,\left(\lambda+a-\frac{1+\mu+\eta}{2}\right)+a\geq a\,,
\end{equation*} 
which together with \eqref{E:Ir.13}, leads to the conclusion that \eqref{E:Ir.10} 
is implied by
\begin{equation*}
\lambda-\frac{1+\mu}{2}\geq \frac{\eta}{2}+
\frac{\sqrt{C}}{a}\,\left[\frac{\big(1+\mu+\eta\big)^2}{4}-\frac14+\eta\right]
+\frac{1-\sqrt{C}}{\sqrt{C}}\,a\,.
\end{equation*}

Since in turn this inequality implies that \eqref{E:Ir.12} holds, we conclude that \textit{Claim 2} is implied by:\\

\textit{Claim 3. Given $\mu\geq0$ and $\lambda>\frac{1+\mu}{2}$, there exist constants $  \eta  \in (0, 2\lambda_v)$,
$a\geq\lambda_v$, and $C\in \big[\frac12,1\big)$ such that
\begin{equation}\label{E:Ir.14}
\lambda-\frac{1+\mu}{2}\geq\frac{\eta}{2}+
\frac{\sqrt{C}}{a}\,\left[\frac{\big(1+\mu+\eta\big)^2}{4}-\frac14+\eta\right]
+\frac{1-\sqrt{C}}{\sqrt{C}}\,a\,.
\end{equation}}
But the proof of \textit{Claim 3} is straightforward. Namely, since $\lambda>\frac{1+\mu}{2}$,
there exists 
\begin{equation*}
0<\varepsilon<\lambda-\frac{1+\mu}{2}\,.
\end{equation*}
Then set
\begin{equation}\label{E:Ir.15}
\eta=\min\left\{\frac{2\varepsilon}{3},2\lambda_v \right\}>0\,,
\end{equation}
and with these
\begin{equation}\label{E:Ir.16}
a=\max\left\{\lambda_v,
\frac{3}{\varepsilon}\,\left[\frac{\big(1+\mu+\eta \big)^2}{4}-\frac14+\eta\right]\right\}\,.
\end{equation}
Finally, having chosen $a$, set
\begin{equation}\label{E:Ir.17}
C=\max\left\{\frac12,\left(\frac{a}{a+\frac{\varepsilon}{3}}\right)^2\right\}\in\left[\frac12,1\right)\,.
\end{equation}
Equations \eqref{E:Ir.15}, \eqref{E:Ir.16}, and \eqref{E:Ir.17} are concrete choices which ensure that
each term of the right-hand side of \eqref{E:Ir.14} is at most $\varepsilon/3$, which by the
choice of $\varepsilon$ guarantees that \textit{Claim 3}, and thus also
\textit{Claim 2} and \textit{Claim 1}, all hold.

ii. Now assume that $\Omega$ is convex, and there exists $\delta_v>0$ such that $v(x)=\frac{1}{\delta(x)}$
for all $\delta(x)<\delta_v$. In this case, we wish to apply Theorem~\ref{T:P} i. with the same choices
of scalar and perturbation potentials made above in \eqref{E:Ir.1},. That is, we will check below that
condition \eqref{E:P.1} holds with these assumptions on $\Omega$ and $v$. 

As \eqref{E:P.1} is the same as \eqref{E:P.2} with $C=1$, it is not surprising that the calculations follow
as in the proof of part i., with a few (simplifying) changes. Namely, since $\Omega$ is convex,
we know from \eqref{E:D.16} that here we can choose $h_0(\Omega)=0$. In addition, the form of $v$
means that the function $\ell(x)\equiv 1$, and hence $\nabla\ell(x)\equiv 0$ on $\delta(x)<\delta_v$.
Among other things, this implies that $v\in\mathcal V^1_\mu$ with $\mu=0$, and there is no need to introduce the
small parameter $\eta$. Indeed, the conclusion follows if we can show that,
given $\lambda\geq\frac12$, there exists $a\geq\lambda_v$ such that 
\begin{equation}\label{E:Ir.18}
\left[\frac14+(\lambda+a)^2\right]\mathds1_4+(\lambda+a)B(x)-a^2\mathds1_4\geq0
\end{equation}
with $B(x)$ defined as in \eqref{E:Ir.4}, but here much simpler since $\nabla\ell\equiv0$.
Thus the bound \eqref{E:Ir.5} holds, and in this case reads
\begin{equation}\label{E:Ir.19}
B(x)\geq -\mathds1_4\quad\text{for all }\delta(x)<\delta_v\,.
\end{equation}

So \eqref{E:Ir.18} is implied by 
\begin{equation}\label{E:Ir.20}
\frac14+(\lambda+a)^2-(\lambda+a)-a^2\geq 0\,.
\end{equation}
But
\begin{equation*}
\frac14+(\lambda+a)^2-(\lambda+a)-a^2=\left(\lambda+a-\frac12\right)^2-a^2=
\left(\lambda-\frac12\right)\left(\lambda+2a-\frac12\right)\,,
\end{equation*}
which, given that $a\geq\lambda_v\geq0$, is implied by the condition $\lambda\geq\frac12$, as claimed.
\end{proof}

\section{Confinement of Dirac particles in $d=1$: Weyl limit point/limit circle approach}\label{S:Ddim1}

In the one-dimensional case, one can apply the powerful Weyl limit point/limit circle theory extended to Dirac operators
to obtain, at least in some particular cases, very precise results. We give here two such examples,
following Weidmann \cite{Wie1,Wie2,Wie3}, which will also be relevant in the next subsection. 

Consider the most general one dimensional Dirac operator (see \eqref{E:D.1.2.3})
\begin{equation}\label{E:OD.1}
\bbD=\sigma_2 D+\bbV(x)
\end{equation}
with $\bbV\in C^0\big((a,b);\IC^2\big)$, and $(a,b)\subset\IR$, $-\infty<a<b<\infty$. As can be directly checked, every Hermitian matrix
can be written as a real linear combination of $\sigma_0=\mathds1_2$ and $\sigma_j$, $1\leq j\leq3$,
and so in particular
\begin{equation}\label{E:OD.2}
\bbV(x)=\sum_{j=0}^3 \sigma_j v_j(x)\,,\quad\text{with } v_j(x)\in\IR\,,\,\,\, 0\leq j\leq 3\,.
\end{equation}
Without loss of generality, we can assume that $v_2\equiv 0$. Indeed, for a given function $v_2$, define
\begin{equation*}
\varphi(x)=\int_a^x v_2(\tilde x)\,d\tilde x\,,
\end{equation*}
and the unitary operator $U$ of multiplication with $e^{-i\varphi}$. Then
\begin{equation}\label{E:OD.4}
U^*\big(\sigma_2 D+\sum_{j=0}^3 \sigma_j v_j\big) U=\sigma_2 D +\sigma_1 v_1+\sigma_3 v_3+v_0\mathds 1_2
\end{equation}
where we note that the resulting potential is real-valued:
\begin{equation*}
\bbV=\sigma_1 v_1+\sigma_3 v_3+v_0\mathds 1_2=\bar\bbV\,.
\end{equation*}
By a slight abuse of notation, we will denote below by $\bbD$ the operator on the right-hand side of \eqref{E:OD.4}
\begin{equation}\label{E:OD.5}
\bbD=\sigma_2 D +\sigma_1 v_1+\sigma_3 v_3+v_0\mathds 1_2\,.
\end{equation}
The following lemma shows that the essential self-adjointness of $\bbD$ depends
on the matrix structure of $\bbV$ and not only on its size.

\begin{lemma}\label{L:NES}
If $v_1=v_3=0$, then $\bbD=\sigma_2 D+v_0\mathds 1_2$ is not essentially self-adjoint.
\end{lemma}

\begin{proof}
Let $A$ be the unitary matrix which diagonalizes $\sigma_2$,
\begin{equation}\label{E:OD.6}
A=
\frac{1}{\sqrt2}\,\left(\begin{matrix}
1 & -i \\ -i &1
\end{matrix}\right) \quad\text{with } A\sigma_2 A^*=\sigma_3\,. 
\end{equation}
Then
\begin{equation}\label{E:OD.7}
A\bbD A^*=
\left(\begin{matrix}
-i\frac{d}{dx}+v_0 &0\\ 0 & i\frac{d}{dx}+v_0
\end{matrix}
\right)\,,
\end{equation}
i.e. $A\bbD A^*$ is a direct sum of two symmetric scalar operators which, by the same gauge transformation
above, are unitarily equivalent with $\pm i\frac{d}{dx}$. But while
$\pm i\frac{d}{dx}$ have self-adjoint extensions, they are not essentially self-adjoint  -- 
see, e.g., the example in \cite{RS}, Sections VIII.2 and X.1.
\end{proof}

Recall $\{\sigma_1,\sigma_2\}=\{\sigma_2,\sigma_3\}=0$, and so
$\sigma_1 v_1+\sigma_3 v_3$ exhausts the class of scalar potentials 
as given by condition \eqref{E:S.16}. In what follows, we seek conditions
on $v_1$ and $v_3$ ensuring essential self-adjointness of $\bbD$,
and we will use the extension to 1 dimensional Dirac operators
of the Weyl limit point/limit circle criterion.
We remind the reader that the operator $\bbD$ on $(a,b)$ is 
said to be limit point at $a$ (or at $b$, respectively) iff the equation
$\bbD\bsPsi=0$ has a solution 
which does not belong to $L^2\big((a,a+\delta_0);\IC^2\big)$ (or $L^2\big((b-\delta_0,b);\IC^2\big)$
respectively) for some $\delta_0>0$. Otherwise, the operator is said to be limit circle at the respective
interval endpoint. 
Given this definition, the following holds, see e.g. \cite{Wie1,Wie2,Wie3}:
\begin{theorem}\label{T:W}
The operator $\bbD$ as given by \eqref{E:OD.5} is essentially
self-adjoint if and only if $\bbD$ is limit point at both $a$ and $b$.
\end{theorem}

From Theorem~\ref{T:W} we obtain:
\begin{proposition}\label{P:M}
Let $v_0=v_3=0$, i.e.
\begin{equation}\label{E:OD.8}
\bbD=\sigma_2 D+\sigma_1 v_1\quad\text{with }\mathcal D(\bbD)=C_0^1\big((a,b);\IC^2\big)\,,
\end{equation}
and let $\delta(x)=\min\{x-a,b-x\}$ denote, as usual, the distance to the boundary of the
(spatial) domain $(a,b)\subset\bbR$.
\begin{itemize}
\item[i.] $\bbD$ is essentially self-adjoint if and only if there exists $0<\delta_0<\frac{b-a}{2}$ 
such that
\begin{equation}\label{E:OD.9}
\int_{b-\delta_0}^b e^{2\left|\int_{b-\delta_0}^x v_1(y)\,dy\right|} dx=\infty
\end{equation}
and
\begin{equation}\label{E:OD.10}
\int_a^{a+\delta_0} e^{2\left|\int_x^{a+\delta_0} v_1(y)\,dy\right|} dx=\infty\,.
\end{equation}
\item[ii.] If there exists $0<\delta_0<\frac{b-a}{2}$ such that
\begin{equation}\label{E:OD.11}
\big|v_1(x)\big|\geq\frac{1}{2\delta(x)}\quad\text{for }\delta(x)\leq\delta_0
\end{equation}
then $\bbD$ is essentially self-adjoint.
\item[iii.] If there exists $0<\delta_0<\frac{b-a}{2}$ and $\lambda<\frac12$ such that
\begin{equation}\label{E:OD.12.0}
\big|v_1(x)\big|\leq \frac{\lambda}{\delta(x)}\quad\text{for } x\in(a,a+\delta_0)
\end{equation}
or
\begin{equation}\label{E:OD.12}
\big|v_1(x)\big|\leq \frac{\lambda}{\delta(x)}\quad\text{for } x\in(b-\delta_0,b)
\end{equation}
then $\bbD$ is not essentially self-adjoint.
\end{itemize}
\end{proposition}

\begin{remark}\label{R:OD1}
Note that it follows from Proposition \ref{P:M} iii. that the result in 
Theorem \ref{T:D1S} ii. is optimal, in the sense that the lower bound 
for $\lambda$ cannot be improved.
\end{remark}

\begin{proof}
In view of Theorem~\ref{T:W} we have to decide whether $\bbD$ is or is not limit point
at $a$ and at $b$. We focus on the situation at $b$, with analogous arguments at $a$. 

i. Let $\bsPsi=(\psi_1\,\,\,\psi_2)^T$ be a solution of $\ID\bsPsi=0$ on $(b-\delta_0,b)$. 
Solving the respective ODEs directly shows that $\psi_1$ is proportional to $e^{-g}$ and
$\psi_2$ is proportional to $e^g$, where
\begin{equation}\label{E:OD.14}
g(x)=\int_{b-\delta_0}^x v_1(y)\,dy\,.
\end{equation}
Consider then the following two linearly independent solutions of $\ID\bsPsi=0$:
\begin{equation}\label{E:OD.15}
\bsPsi_1=\left(\begin{matrix} e^{-g} \\ 0
\end{matrix}\right)\quad\text{and}\quad
\bsPsi_2=\left(\begin{matrix} 0 \\ e^g
\end{matrix}\right)\,.
\end{equation}
Note that
\begin{equation*}
\int_{b-\delta_0}^b \big|\bsPsi_1(x)\big|^2+\big|\bsPsi_2(x)\big|^2\,dx
=\int_{b-\delta_0}^b e^{2g(x)}+e^{-2g(x)}\,dx\geq \int_{b-\delta_0}^b e^{2|g(x)|}\,dx\,,
\end{equation*}
and so \eqref{E:OD.9} implies that it cannot be that both $\bsPsi_1$ and $\bsPsi_2$
are in $L^2\big((b-\delta_0,b);\IC^2\big)$. On the other hand,
$$
\int_{b-\delta_0}^b e^{2g(x)}+e^{-2g(x)}\,dx\leq 2\int_{b-\delta_0}^b e^{2|g(x)|}\,dx\,,
$$
so if \eqref{E:OD.9} does not hold, then $\bsPsi_1,\bsPsi_2\in \big(L^2(b-\delta_0,b)\big)^2$,
which in turn implies that all solutions of $\ID\bsPsi=0$ are in $\big(L^2(b-\delta_0,b)\big)^2$.
We have thus concluded that $\ID$ is limit point at$b$ iff \eqref{E:OD.9} holds.

ii. Since $v_1$ is a continuous function, \eqref{E:OD.11} implies that it has constant sign, and hence
\begin{equation*}
e^{2\left|\int_{b-\delta_0}^x v_1(y)\,dy\right|}=e^{2\int_{b-\delta_0}^x |v_1(y)|\,dy}
\geq \frac{\delta_0}{b-x}\,,
\end{equation*}
which directly guarantees that \eqref{E:OD.9} holds and so $\ID$ is 
limit point at $b$.

iii. Assume that \eqref{E:OD.12} holds. Then, with the notation \eqref{E:OD.14},
\begin{equation*}
e^{2\big|\int_{b-\delta_0}^x v_1(y)\,dy\big|} \leq e^{2\int_{b-\delta_0}^x|v_1(y)|\,dy}\leq \left(\frac{\delta_0}{b-x}\right)^{2\lambda}\,.
\end{equation*}
As $\lambda<\frac12$, this implies that \eqref{E:OD.9} does not hold, and so
$\ID$ is limit circle at $b$.
\end{proof}

Finally, we treat the general case:
\begin{corollary}\label{C:SMF} 
Consider
\begin{equation}\label{E:OD.17}
\ID=\sigma_2 D+\sigma_1 v_1+\sigma_3 v_3+ v_0\mathds 1\,,
\end{equation}
with
\begin{equation}\label{E:OD.16}
v_j(x)=\frac{\lambda_j}{\delta(x)}\quad\text{for }\delta(x)\text{ small enough and }j=0,1,3\,.
\end{equation}
Then $\ID$ is essentially self-adjoint if and only if
\begin{equation}\label{E:OD.18}
\lambda_0^2\leq \lambda_1^2+\lambda_3^2-\tfrac14\,.
\end{equation}
\end{corollary}

The proof of this statement closely mimics the proof of Theorem~6.9 in \cite{Wie2},
and is left to the interested reader.

\section{Confinement of Dirac particles in $d=2$: magnetic fields}\label{S:Ddim2}

We consider now the question of confinement of relativistic particles with spin 1/2 (Dirac particles)
solely by magnetic fields. Since magnetic potentials do not satisfy Definition~\ref{D:ScaPo}, i.e. are
not scalar, a general theory of purely magnetic confinement does not exist (assuming such confinement 
is even possible). Even for nonrelativistic spinless particles, positive results in a general setting have only been 
obtained recently by Y. Colin de Verdi\`ere and F. Truc \cite{CdVT}. As for nonrelativistic particles with spin 1/2,
confinement was proved \cite{NN2} only for the unit disc in $\IR^2$ and rotationally invariant magnetic fields.

However, the situation is much better in 2 dimensions. More precisely, we show below that for bounded
domains $\Omega$ in $\IR^2$, the magnetic Dirac operator is essentially self-adjoint on $C_0^\infty\big(\Omega;\IC^2\big)$
provided the magnetic field $\cB$ satisfies a simple growth condition near $\partial\Omega$, see Theorem~\ref{T:M2} below.
The proof rests on the supersymmetric structure of the Dirac operator in 2 dimensions (see \cite[Section 7.1.2]{Th}), 
which allows for the reduction of this problem to the essential self-adjointness of a scalar magnetic Schr\"odinger operator. 
This in turn is amenable to the method of \cite{NN1,NN3} combined with the diamagnetic inequality in \cite{CdVT}.

The magnetic Dirac operator in 2 dimensions is given by
\begin{equation}\label{E:M.2}
\ID_{2,mag}=\sigma_1\big(D_1-\cA_1\big)+\sigma_2\big(D_2-\cA_2\big)
\end{equation}
where 
\begin{equation}\label{E:M.3}
\bscA=\big(\cA_1\,,\,\cA_2\big)\in C^1\big(\Om;\IR^2\big)
\end{equation}
is the magnetic vector potential. 
We again drop the mass term $\sigma_3 m$ from the standard Dirac operator, since
it is uniformly bounded and hence irrelevant for essential self-adjointness. 

Recall that, even though $\bscA$ appears in \eqref{E:M.2}, the physically relevant quantity is the magnetic field 
\begin{equation}\label{E:M.4}
\cB(x)=\frac{\partial \cA_2}{\partial x_1}(x)-\frac{\partial \cA_1}{\partial x_2}(x)
\end{equation}
The main result of this section is then the following:
\begin{theorem}\label{T:M2}
Assume that there exists $\delta_0>0$ such that either
\begin{equation}\label{E:TM2.1}
\cB(x)\geq\frac12\cdot\frac{1}{\delta(x)^2}\quad\text{for all }x\in\Om\text{ with }
\delta(x)<\delta_0
\end{equation}
or
\begin{equation}\label{E:TM2.2}
\cB(x)\leq-\frac12\cdot\frac{1}{\delta(x)^2}\quad\text{for all }x\in\Om\text{ with }
\delta(x)<\delta_0\,,
\end{equation}
where $\delta(x)=\text{dist}(x,\partial\Om)$.
Then $\ID_{2,mag}$ is essentially self-adjoint on $C_0^2\big(\Omega;\IC^2\big)$.
\end{theorem}

\begin{remark}
Note that, if $\Om$ is also simply connected, then the condition that either \eqref{E:TM2.1}
or \eqref{E:TM2.2} holds
is equivalent to
\begin{equation}\label{E:TM2.3}
\big|\cB(x)\big|\geq\frac12\cdot\frac{1}{\delta(x)^2}\quad\text{for all }x\in\Om\text{ with }
\delta(x)<\delta_0\,,
\end{equation}
since $\partial\Om$ is connected and $\cB$ is continuous.
\end{remark}

As already discussed, the main ingredient of the proof of this theorem is the abstract ``supersymmetry'' lemma of 
F. Gesztesy, B. Simon, and B. Thaller in \cite{GST} (see also Lemma~ 5.7 in \cite{Th}).
\begin{lemma}\label{L:SU}
Let $\mathcal H_\pm$ be separable Hilbert spaces, 
$$
D_\pm\,:\,\mathcal D_\pm\subset\mathcal H_\pm\,\rightarrow\, \mathcal H_\pm
$$
densely defined closable operators, $\cH=\cH_+\oplus\cH_-$, and
\begin{equation}\label{E:M.5}
D=\left(
\begin{matrix}
0 & D_-\\ D_+ &0
\end{matrix}\right)
\qquad\text{with}\quad \cD om(D)=\cD_+\oplus\cD_-\,.
\end{equation}
Assume that $D$ is symmetric, and that
\begin{itemize}
\item[(i)] $D_+(\cD_+)\subset\cD om(D_-^{**})$
and $D_-^{**}D_+$ is essentially self-adjoint on $\cD_+$, \label{E:M.6}
\item[\textit{or}] 
\item[(ii)] $D_-(\cD_-)\subset\cD om(D_+^{**})$ and $D_+^{**}D_-$ is essentially self-adjoint on $\cD_-$. \label{E:M.6.1}
\end{itemize}

Then $D$ is essentially self-adjoint.
\end{lemma}

\begin{proof}[Proof of Theorem~\ref{T:M2}]
We start by noting that, due to the explicit, off-diagonal form of $\sigma_1$ and $\sigma_2$,
we can write our operator of interest as
\begin{equation}\label{E:M.7}
\ID_{2,mag}=
\left(\begin{matrix}
0 & D_1-iD_2-\cA_1+i\cA_2 \\
D_1+iD_2-\cA_1-i\cA_2 & 0
\end{matrix}\right)
\end{equation}
i.e. in the form \eqref{E:M.5} with
\begin{equation}\label{E:M.7.1}
D_\pm=D_1\pm iD_2-\cA_1\mp i\cA_2\,,\quad \cH_\pm=L^2(\Om)\,,\quad \cD_\pm=C_0^2(\Om)\,.
\end{equation}
It is clear that $D_\pm$ are densely defined, and since $D_\mp$ is the formal adjoint of $D_\pm$,
it is easily seen that $\ID_{2,mag}$ is symmetric. It is also straightforward to check that
$C_0^1(\Om)\subset\mathcal{D}om(D_\pm^*)$, and hence (see, for example, \cite[Theorem 1.8]{Sch})
$D_\pm$ are closable and $D_\pm^{**}=\overline{D_\pm}$. It follows that 
$\cD om(D_-^{**})=\cD om (\overline{D_-})\supset C_0^1(\Om)$ which together with 
$D_+(C_0^2(\Om))\subset C_0^1(\Om)$ gives $D_+(\cD_+)\subset\cD om(D_-^{**})$. 
Furthermore, a direct computation shows that on $C_0^2(\Om)$
\begin{equation}\label{E:M.8}
D_-^{**}D_+=(D_1-\cA_1)^2+(D_2-\cA_2)^2-\cB\,.
\end{equation}
Interchanging the roles of $D_+$ and $D_-$ one obtains that $D_-(\cD_-)\subset\cD om(D_+^{**})$
and that on $C_0^2(\Om)$
\begin{equation}\label{E:M.9}
D_+^{**}D_-=(D_1-\cA_1)^2+(D_2-\cA_2)^2+\cB\,.
\end{equation}
Using Lemma~\ref{L:SU}, we can then conclude essential self-adjointness of $\ID_{2,mag}$
on $C_0^2\big(\Om;\IC^2\big)$ if (at least) one of 
\begin{equation}\label{E:M.10}
H_{_{\cA,\pm}}=(D_1-\cA_1)^2+(D_2-\cA_2)^2\pm\cB\qquad\text{with}\quad \cD om(H_{_{\cA,\pm}})=C_0^2(\Om)
\end{equation}
is essentially self-adjoint.

We can assume,
without loss of generality, that $\cB(x)>0$ for $\delta(x)<\delta_0$ and focus on proving
the essential self-adjointness of $H_{_{\cA,+}}$.
If, instead, $\cB(x)<0$ for $\delta(x)<\delta_0$, then the proof below yields essential self-adjointness for
$H_{_{\cA,-}}$.
The exponential Agmon estimates method from \cite{NN1} as applied in \cite{CdVT} and \cite{NN3}
gives essential self-adjointness of $H_{_{\cA,+}}$ provided there exists a function $h$ with $|h|$
uniformly bounded on compacts in $\Om$, and such that 
\begin{equation}\label{E:M.12}
h(x)\geq\frac{1}{\delta(x)^2}\quad\text{for $\delta(x)$ sufficiently small,}
\end{equation}
and 
\begin{equation}\label{E:M.11}
\big(\varphi\,,\,H_{_{\cA,+}}\varphi\big)\geq \int_\Om h(x)\big|\varphi(x)\big|^2\,dx\quad\text{for all }\varphi\in C_0^2(\Om)\,.
\end{equation}

Note that we can, at this point, conclude essential self-adjointness of $H_{_{\cA,+}}$ (and hence $\ID_{2,mag}$) 
if $\cB(x)\geq \delta(x)^{-2}$ near
$\partial\Omega$, but this condition misses the claimed \eqref{E:TM2.1} by a $\frac12$ factor. To recover essential 
self-adjointness based only on \eqref{E:TM2.1} we must use the following
elementary diamagnetic inequality (see \cite{CdVT}):
\begin{equation}\label{E:M.13}
\Big(\varphi\,,\, \left((D_1-\cA_1)^2+(D_2-\cA_2)^2\right)\varphi\Big)\geq \left|\int_\Om \cB(x)\big|\varphi(x)\big|^2\,dx\right|
\end{equation}
This follows from the fact that 
\begin{equation*}
\big[D_1-\cA_1,D_2-\cA_2\big]=i\cB\,,
\end{equation*}
and so
\begin{equation*}
\begin{aligned}
\big|\big(\varphi,\cB \varphi\big)\big| \leq 2\big\|(D_1-\cA_1)\varphi\big\| \cdot\big\|(D_2-\cA_2)\varphi\big\|
\leq \big\|(D_1-\cA_1)\varphi\big\|^2+\big\|(D_2-\cA_2)\varphi\big\|^2\,.
\end{aligned}
\end{equation*}
We now turn to the operator $H_{_{\cA,+}}$ from \eqref{E:M.10}, and recall that we have assumed that
there exists $\delta_0>0$ such that
$$
\cB(x)\geq\frac{1}{2\delta(x)^2}\qquad\text{for } \delta(x)<\delta_0\,.
$$
We choose the function
\begin{equation*}
h(x)=
\begin{cases}
2\cB(x) &\quad\text{for }\delta(x)<\delta_0\\
-\sup_{\delta(x)\geq\delta_0} 2\big|\cB(x)\big| &\quad\text{for } \delta(x)\geq\delta_0\,.
\end{cases}
\end{equation*}
Then \eqref{E:TM2.1} and \eqref{E:M.13} imply that \eqref{E:M.11} holds for this $h$, and hence
the operator $H_{_{\cA,+}}$ is essentially self-adjoint, as claimed.
\end{proof}

We now turn to the proof of the optimality of condition \eqref{E:TM2.1}, which we achieve
by considering a special problem on the unit disk in $\IR^2$. More precisely, 
for the remainder of this section we fix
\begin{equation*}
\Om_{_{disk}}=\big\{x=(x_1,x_2)\in\IR^2\,\big|\, |x|^2=x_1^2+x_2^2<1\big\}
\end{equation*}
on which we consider
a rotationally symmetric magnetic field $\cB$. As is standard, we use a slight abuse of notation to write:
\begin{equation}\label{E:M.14}
\cB(x)=\cB(|x|)\qquad \text{for }x\in\Omd\,,
\end{equation}
where $\cB$ is now a continuous function of $r=|x|\in[0,1)$.

Recall that the magnetic field $\cB$ does not uniquely determine the magnetic potential $\bscA$. 
However, if the physical domain $\Om$ is simply connected, then any two
magnetic potentials $\bscA$ and $\bscA^\prime$ associated, via \eqref{E:M.4}, to the same magnetic field $\cB$
differ by a gradient: i.e., there exists $\varphi\in C^2(\Om;\IR)$ such that  
$\bscA^\prime\equiv\bscA_\varphi=\bscA+\nabla\varphi$. This in particular implies that,
even though the corresponding Dirac Hamiltonians are different, they are
unitarily equivalent:
\begin{equation}\label{E:M.15}
U_\varphi^*\big(\bssig\cdot(\bsD-\bscA)\big) U_\varphi=\bssig\cdot(\bsD-\bscA_\varphi)\,,
\end{equation}
where $U_\varphi$ is the unitary operator of multiplication with $e^{i\varphi}$.
$U_\varphi$ is called a gauge transformation, and \eqref{E:M.15} is gauge covariance.
The fact that essential self-adjointness is stable under gauge transformations
follows from the invariance of $C_0^1\big(\Om;\IC^2\big)$ under $U_\varphi$
and the stability of essential self-adjointness under unitary transformations 
that preserve the domain. A particular choice of a magnetic potential $\bscA$ associated,
via \eqref{E:M.4}, to the magnetic field $\cB$ is called gauge fixing, and 
gauge covariance implies that all physically relevant results obtained for $\bscA$ remain valid
for all other $\bscA_\varphi$. This freedom of gauge fixing is particularly useful at the technical level,
as it allows one to choose a well-suited gauge for each problem.

In our case, we choose to work with the transversal (aka Poincar\'e) gauge, which is characterized by
\begin{equation}\label{E:M.16}
\bscA(x)\perp x\quad\text{for all }x\in\Omd\,.
\end{equation}
Using polar coordinates $r$ and $\theta$ and the notations (see, for example, \cite[Section 7.3.3]{Th})
\begin{equation}\label{E:M.18}
\boldsymbol{e}_\theta=\tfrac1r\,\big(-x_2,x_1\big)\quad\text{and}\quad \boldsymbol{e}_r=\tfrac1r\,\big(x_1,x_2\big)\,,
\end{equation}
the transversality condition \eqref{E:M.16} amounts to writing the magnetic potential as
\begin{equation}\label{E:M.19}
\bscA(x)=a(r)\boldsymbol{e}_\theta
\end{equation}
with $a(0)=0$.
A straightforward calculation then shows that
\begin{equation}\label{E:M.20}
\cB(r)=\tfrac1r a(r)+a'(r)=\tfrac1r\,\tfrac{d}{dr}\big(ra(r)\big)
\end{equation}
and
\begin{equation}\label{E:M.21}
a(r)=\frac1r\int_0^r y\cB(y)\,dy\,.
\end{equation}

To fully utilize the polar coordinates in this context, we must consider
\begin{equation}\label{E:M.22}
\oID_{2,mag}=\ID_{2,mag}\Big|_{C_0^1\big(\oOmd;\IC^2\big)}\,,
\end{equation}
where
\begin{equation}\label{E:M.22.1}
\oOmd=\Omd\setminus\{0\}=\big\{x=(x_1,x_2)\in\IR^2\,\big|\, 0<|x|^2=x_1^2+x_2^2<1\big\}\,.
\end{equation}
Since $\ID_{2,mag}$ is a symmetric extension of $\oID_{2,mag}$, the essential
self-adjointness of $\oID_{2,mag}$ directly implies that of $\ID_{2,mag}$. The converse 
follows from the following proposition, which is of independent interest.

\begin{proposition}\label{L:P}
Let $\Om\subset\IR^2$ be a domain and $x_0\in\Om$. Consider 
the Dirac operator on $\Om$
$$
\ID=\bssig\cdot\bsD+\bbV\qquad \cD om(\ID)=C_0^1\big(\Om;\IC^2\big)\,,
$$
where we assume that $\bbV=\bbV^*\in L^\infty_{loc}(\Om;\IC^{2\times2})$. If $\ID$ is essentially self-adjoint, 
then so is
$$
\oID=\ID\Big|_{C_0^1\big((\Om\setminus\{x_0\};\IC^2\big)}\,.
$$
\end{proposition}

\begin{proof}
First observe that, without loss of generality, we can assume that $x_0=0$. We then proceed in three main steps.

\textbf{Step 1.} We start with the case where $\Om=\IR^2$ and the operator is $\ID_0=\bssig\cdot\bsD$ 
is the free Dirac operator, with $\cD om(\ID_0)=C_0^1\big(\IR^2;\IC^2\big)$. In this case, it is already 
known that $\ID_0$ is essentially self-adjoint on $C_0^1\big(\IR^2\setminus\{0\};\IC^2\big)$, 
which is the claim on the lemma. However, for completeness, we include here a proof which mimics 
the proof in \cite[Thm. X.11]{RS} of the fact that $-\Delta$ is essentially self-adjoint on $C_0^2(\IR^n\setminus\{0\})$
for $n\geq 4$.

The same type of arguments as those leading to Lemma~\ref{L:PW} show that $\ID_0$ is essentially self-adjoint on 
$C_0^1\big(\IR^2\setminus\{0\};\IC^2\big)$ iff 
\begin{equation}\label{E:A.1}
\ID_{0,m_j}=\sigma_2 D_r-\sigma_1\frac{m_j}{r}\quad\text{with }\,\,\,m_j=\frac{2j-1}{2}\,,\,\,\, j\in\bbZ
\end{equation}
is essentially self adjoint on $C_0^1\big((0,\infty);\IC^2\big)$ for every $j\in\bbZ$.
Since, for all $j\in\bbZ$, $|m_j|\geq\frac12$, we see that $\ID_{0,m_j}$ is limit-point at both
0 and $\infty$, and thus essentially self-adjoint, as claimed (see, for example, \cite{Wie2}).

\textbf{Step 2.} The second step aims to prove a technical approximation result which will then allow us,
in Step 3, to localize from $\IR^2$ to an arbitrary $\Om$. 

We start by recalling some elementary facts about
extensions of symmetric operators. 
Let $S$ and $T$ be symmetric operators, and $T\subset S$. Then $\overline{T}\subset\overline{S}$,
where $\overline{S}$ and $\overline{T}$ are the closures of $S$ and $T$, respectively. Assume that
$$
\overline{T}=T^*\,,
$$
that is $T$ is essentially self-adjoint. From the maximality of self-adjoint operators (see, for example,
\cite[Section 3.2]{Sch}) it follows that
\begin{equation}\label{E:A.2}
\overline{S}=\overline{T}=T^*=S^*\,.
\end{equation}
Taking $S=\ID_0$ and $T=\oID_0=\ID_0\Big|_{C_0^1(\IR^2\setminus\{0\};\IC^2)}$, one in particular obtains that
\begin{equation}\label{E:A.3}
C_0^1\big(\IR^2;\IC^2\big)=\cD om(\ID_0)\subset\cD om\big(\,\overline{\ID_0}\,\big)
=\cD om\bigg(\,\overline{\oID_0}\,\bigg)\,.  
\end{equation}

Now let $a\in(0,\infty)$ and $\bsPhi\in C_0^1\big(\IR^2;\IC^2\big)$ with $\supp(\bsPhi)\subset
\big\{x\in\IR^2\,\big|\,|x|\leq a\big\}$. From \eqref{E:A.3} it follows that there exists a sequence
$\big(\bsPhint\big)_{n\geq 1}\in C_0^1\big(\IR^2\setminus\{0\};\IC^2\big)$ such that 
\begin{equation}\label{E:A.5}
\bsPhint \to \bsPhi\quad\text{and}\quad \ID_0\bsPhint \to \ID_0\bsPhi\quad
\text{in }L^2\big(\IR^2;\IC^2\big)\,.
\end{equation}
Let $\chi\in C_0^1(\IR^2)$ be a smooth cut-off function such that $\chi(x)=1$ for $|x|\leq a$ and
$\chi(x)=0$ for $|x|\geq 2a$. For each $n\geq 1$, define $\bsPhin=\chi\bsPhint$, which together
with the fact that $\supp(\bsPhi)\subset \big\{x\in\IR^2\,\big|\,|x|\leq a\big\}$ then implies that
\begin{equation}\label{E:A.6}
\bsPhin-\bsPhi=\chi\big(\bsPhint-\bsPhi\big)
\end{equation}
This in turn shows that
\begin{equation*}
\ID_0\bsPhin-\ID_0\bsPhi=\big[\ID_0,\chi\big](\bsPhint-\bsPhi)+\chi\big(\ID_0\bsPhint - \ID_0\bsPhi\big)\,.
\end{equation*}
Together with \eqref{E:A.5}, we can then conclude that the sequence 
\begin{equation}\label{E:A.4.1}
\big(\bsPhin\big)_{n\geq 1}\subset C_0^1\big(\{0<|x|\leq 2a\};\IC^2\big)
\end{equation}
satisfies
\begin{equation}\label{E:A.4}
\bsPhin \to \bsPhi\quad\text{and}\quad \ID_0\bsPhin \to \ID_0\bsPhi\quad
\text{in }L^2\big(\IR^2;\IC^2\big)\,.
\end{equation}

\textbf{Step 3.} We now turn to the general case of the operator $\ID=\bssig\cdot\bsD+\bbV$ on a 
general open connected set $\Om\subset\IR^2$ with $0\in\Om$. 

We first claim that
\begin{equation}\label{E:A.7}
C_0^1\big(\Om;\IC^2\big)\subset \cD om\big(\,\overline{\oID}\,\big)
\end{equation}
where we recall that $\oID=\ID\Big|_{C_0^1\big(\Om\setminus\{0\};\IC^2\big)}$. Indeed, let
$\bsPsi\in C_0^1\big(\Om;\IC^2\big)$. Since $0\in\Om$ open set, there exists $a>0$
such that $\big\{x\in\IR^2\,\big|\, |x|\leq 3a\big\}\subset\Om$. With the same choice of $\chi$ as in Step 2,
we can apply \eqref{E:A.4.1}--\eqref{E:A.4} to the function $\chi\bsPsi$ to show that there exists a sequence
\begin{equation}\label{E:A.8.1}
\big(\bsPhin\big)_{n\geq 1}\subset C_0^1\big(\{0<|x|\leq 2a\};\IC^2\big)
\end{equation}
such that
\begin{equation}\label{E:A.8}
\bsPhin \to \chi\bsPsi\quad\text{and}\quad \ID_0\bsPhin \to \ID_0\chi\bsPsi\quad\text{in }L^2\big(\IR^2;\IC^2\big)\,.
\end{equation}

Then consider the sequence given by
\begin{equation}\label{E:A.9}
\bsPhinp=\bsPhin+(1-\chi)\bsPsi\quad\text{for all }n\geq1\,.
\end{equation}
By construction,
\begin{equation}\label{E:A.9.1}
\big(\bsPhinp\big)_{n\geq 1}\subset \big(C_0^1(\IR^2\setminus\{0\})\big)^2=\cD om(\oID)\quad
\text{and}\quad \bsPhinp\to\bsPsi\,.
\end{equation}
Furthermore, recall that $\bbV\in L^\infty_{loc}(\Om)$, and so \eqref{E:A.8.1}--\eqref{E:A.9} yield that
\begin{equation*}
\oID\bsPhinp= \ID_0\bsPhin+\bbV\bsPhin+\ID(1-\chi)\bsPhi\to 
\ID_0\chi\bsPhi+\bbV\chi\bsPhi+\ID(1-\chi)\bsPhi=\ID\bsPsi\,.
\end{equation*}
In other words, $\bsPsi\in\cD om\big(\,\overline{\oID}\,\big)$, which, since $\bsPsi$ was arbitrary,
proves \eqref{E:A.7}, as claimed. 

Finally, note that \eqref{E:A.7} translates to 
\begin{equation*}
\ID\subset \overline{\oID}\,.
\end{equation*}
By assumption, $\overline\ID$ is self-adjoint, so the conclusion of Lemma~\ref{L:PW} follows
from the maximality property of self-adjoint operators. 
\end{proof}

Focusing now on $\oID_{2,mag}$, its essential self-adjointness can be investigated by partial wave analysis
(see, for example, \cite[Section 7.3.3]{Th}). Indeed, for each $j\in\bbZ$ consider the operator
\begin{equation}\label{E:M.23}
\ID_{m_j}=\sigma_2 D_r+\sigma_1\left(a(r)-\tfrac{m_j}{r}\right)\quad
\text{on }\cD om(\ID_{m_j})=C_0^1\big((0,1);\IC^2\big)\,,
\end{equation}
where
\begin{equation}\label{E:M.23.1}
D_r=-i\tfrac{d}{dr}\quad\text{and}\quad m_j=\tfrac{2j+1}{2}\,.
\end{equation}
One then knows that:
\begin{lemma}\label{L:PW}
$\oID_{2,mag}$ is essentially self-adjoint iff $ \ID_{m_j}$ is essentially self-adjoint for every $j\in\bbZ$.
\end{lemma}

We can then use all of the above to conclude that the constant $1/2$ in \eqref{E:TM2.1} is optimal:
\begin{proposition}\label{P:CM}
Let $\Omd=\big\{x=(x_1,x_2)\in\IR^2\,\big|\, |x|<1\big\}$ be the open unit disk in $\IR^2$,
\begin{equation}\label{E:M.24}
\cB(r)=\frac{\alpha}{(1-r)^2}\quad\text{for } r\in[0,1)\,,
\end{equation}
and $\bscA$ the associated transversal gauge magnetic potential as given by \eqref{E:M.19}
and \eqref{E:M.21}.

If $\alpha\in\big[0,\tfrac12\big)$, then $\ID_{2,mag}=\bssig\cdot\big(\bsD-\bscA\big)$ is not essentially self-adjoint on $C_0^1\big(\Omd;\IC^2\big)$.
\end{proposition} 

\begin{proof}
Assume, by contradiction, that $\ID_{2,mag}$ is essentially self-adjoint. As explained above, this implies that
$\oID_{2,mag}$, and hence all $\ID_{m_j}$ for $j\in\bbZ$ are also essentially self-adjoint. 

Now focus, for definitiveness, on $j=-1$:
\begin{equation}\label{E:M.25}
\ID_{-1/2}=\sigma_2 D_r+\sigma_1\big(a(r)+\tfrac{1}{2r}\big)\,.
\end{equation}
From \eqref{E:M.21} and recalling that $r<1$, we obtain that
\begin{equation*}
\begin{aligned}
0<a(r)+\frac{1}{2r}&=\frac{1}{r}\left[\int_0^r \frac{\alpha y}{(1-y)^2}\,dy+\frac12\right] 
\leq\frac{1}{r}\left[\int_0^r \frac{\alpha }{(1-y)^2}\,dy+\frac12\right]\\
&\leq \frac{2\alpha+1-r}{2r}\cdot\frac{1}{1-r}\,.
\end{aligned}
\end{equation*}

Note that
$$
\lim_{r\to1-} \frac{2\alpha+1-r}{2r}=\alpha
$$
and that, by simply taking the average of $\alpha$ and $\frac12$, we have
$$
\alpha<\frac{2\alpha+1}{4}<\frac12\,.
$$
It follows that there exists $\delta_0>0$ such that 
$$
0<\frac{2\alpha+1-r}{2r} \leq \frac{2\alpha+1}{4}<\frac12\quad\text{for }r\in(1-\delta_0,1)\,.
$$
In other words, condition \eqref{E:OD.12} is satisfied, and hence Proposition~\ref{P:M}(iii) implies that
$\ID_{-1/2}$ is not essentially self-adjoint, which is a contradiction.
\end{proof}

\section{Comments and open problems}\label{S:Comments}

In this section we comment on certain extensions of our results, on further connections with
previous results, as well as on open questions.

\textbf{1.} We start by discussing the smoothness hypotheses we made in
Sections 2 and 3 for $\bsAj$ and $\bbV_0$. 
We required $\bsAj\in C^1\big(\Om; \C^{k \times k}\big)$ and $\bbV\in C^0\big(\Om; \C^{k \times k}\big)$,
but these conditions can be weakened to $\bsAj$ Lipschitz continuous and 
$\bbV\in L^\infty_{loc}\big(\Om; \C^{k \times k}\big)$.
Note that Garofalo and Nhieu already mentioned in \cite[Appendix]{GN} that the main technical result
behind \cite[Main Theorem]{Fr} remains true for $\bsAj$ locally Lipschitz (see also \cite[Proposition~A.3]{BMS}).

\textbf{2.} It is also possible to extend our results to unbounded domains. An easy case is when $\partial\Om$
is compact, i.e. when there exists $R>0$ such that $\big\{x\in\IR^d\,\big|\,|x|>R\big\}\subset\Om$,
and when, for all $j$ and for all $x$ sufficiently large,
\begin{equation}\label{E:CO.1}
\big|\bsAj(x)\big|\leq\rho\big(|x|\big)\quad\text{with }\int^\infty \frac{dr}{\rho(r)}=\infty\,.
\end{equation}
This holds, for example, if all the $\bsAj$ are bounded at $\infty$. Other unbounded cases can be considered,
but they require a more careful choice of the function $g$ in the basic inequality from Lemma~\ref{L:LemmaB}.
(See, for example, the choice of functions $\phi$ in the proof of Theorem 3.1 in \cite{NN3}).  

\textbf{3.} In addition to the classical results of Chernoff \cite{Ch1,Ch2} and the results in \cite{NN4} for the smooth case, 
criteria for essential self-adjointness of general first order matrix-valued differential operators 
with rough coefficients on $\IR^d$ were recently obtained by completely different methods in \cite{CGKLNPS},
and in \cite{BS} in the elliptic case. 

In \cite{CGKLNPS}, the authors consider abstract operators which, in our case, have the form
\begin{equation*}
\ID=\sum_{j=1}^d \bsAj D_j+D_j\bsAj^*\,,
\end{equation*}
These operators 
can be rewritten in symmetric form as
\begin{equation*}
\ID=\sum_{j=1}^d \left(\tfrac12\big(\bsAj+\bsAj^*\big)D_j+D_j \,\tfrac12\big(\bsAj+\bsAj^*\big)\right)+ \tfrac12D_j(\bsAj^*)-\tfrac12 D_j(\bsAj)\,.
\end{equation*}
Hypothesis~2.1 of \cite{CGKLNPS} specialized to this case assumes that the
$\bsAj$'s and all their partial derivatives are Lipschitz and bounded as $|x|\to\infty$. So if
the arguments from our Comments~1 and 2 above are correct, then the extension
of our results to the unbounded $\Om$ case covers this case of Theorem~2.4 in \cite{CGKLNPS}.

Turning to \cite{BS},  it was surprising to read the authors' claim that, in the elliptic case, one can
dispense with the Chernoff-type condition \eqref{E:CO.1} (see \cite[Remark 3.11]{BS}). In our opinion, 
such a result cannot hold true, and upon closer inspection, 
we found an error in their main estimate of Section 6.1 which, when corrected, leads
to the Chernoff-type condition being necessary.

In addition, the following example shows that (even uniform) ellipticity cannot replace 
the Chernoff condition. See also \cite{Fa} for another example concerning the optimality of the Chernoff condition.
Let $\alpha \geq 0$ and define the function 
$$
a_\alpha\,:\,\IR\to\IR\qquad a_\alpha(x)=\big(1+x^2\big)^{\alpha/2}\,.
$$
Note that $a_\alpha\in C^\infty(\IR)$ and $a_\alpha(x)\geq 1>0$ for all $x\in\IR$.
The operator
$\ID_{\alpha}=a_\alpha D+Da_\alpha$ with $\cD om(\ID_\alpha)=C_0^1(\IR)$
is thus symmetric in $L^2(\mathbb R)$, and uniformly elliptic. The solutions of $(\ID_{\alpha} \pm i)\Psi_\pm=0$  
can be computed explicitly as 
$$
\Psi_\pm(x)^2=\frac{C_{\pm}}{a_\alpha(x)}e^{\pm\int_0^x \frac{dy}{a_\alpha(y)}}\,.
$$
It is then straightforward to show that $\Psi_\pm\in L^2(\mathbb R)$ if and only if $\alpha >1$. By 
the fundamental criterion for essential self-adjointness, we conclude that 
$\ID_{\alpha}$ is essential self-adjoint 
if and only if $\alpha \leq 1$, which fits exactly with \eqref{E:CO.1}.

\textbf{4.} A case we have not discussed at all so far is $\Om=\bbR^3\setminus\{0\}$. However,
this case has already been thoroughly studied in the context of Dirac operators with Coulomb-type 
singularities. For a detailed discussion, see for example \cite{Th} and \cite{ADV,CP,EL}, and the references therein.

\textbf{5.} For simplicity, and because of their physical relevance, we only considered Lorentz scalar potentials $\bbV_{Ls}$, as given by \eqref{E:D.8}, in our analysis in Section~\ref{S:Dscal}. 
It is a natural problem to extend those results to the class $\cS_s$ of scalar potentials described by Definition~\ref{D:ScaPo}.

For dimension $d=1$, the results in Section~\ref{S:Ddim1} give a satisfactory answer to this question, 
as the class $\mathcal S_s$ is exhausted by potentials of the form
$v_1\sigma_1+v_3\sigma_3$, for which Proposition~\ref{P:M} and Corollary~\ref{C:SMF} give optimal essential self-adjointness results.
In the 2-dimensional case the general form of $\bbV$ can be written as (see \cite{Th}):
$$
\bbV(x)=\sum_{j=1}^3 v_j(x)\sigma_j+v_0(x)\mathds 1_2\,, \quad v_j=\bar v_j\,,\,\,\, j=0,1,2,3\,.
$$
The conditions $\big\{\bbV(x),\sigma_1\big\}=\big\{\bbV(x),\sigma_2\big\}=0$
imply that $v_1=v_2=v_0=0$, and thus the class $\cS_s$ coincides with the class $\cS_{Ls}$, 
for which one has the essential self-adjointness criteria given in Section~\ref{S:Dscal}.

In dimension $d=3$ the general form of $\bbV$ writes as (see \cite{Th}):
$$
\bbV(x)=\sum_{\ell=1}^{16} v_\ell(x)\,\Gamma_\ell\,,\quad v_\ell=\bar v_\ell\,,
$$
where the (constant) $4\times4$ matrices $\Gamma_\ell$ are listed in \cite[Appendix to Ch. 2]{Th}. For example, $\Gamma_2=\beta$
and
$$
\Gamma_{12}=\left(\begin{matrix}
0 & -i\mathds 1_2 \\ i\mathds 1_2 & 0\end{matrix}\right)\,.
$$
A straightforward calculation shows that
\begin{equation}\label{E:star}
\Gamma_{12}^2=\mathds 1_4\quad\text{and}\quad \big\{\Gamma_{12},\beta\big\}= \big\{\Gamma_{12},\alpha_j\big\}=0\,,\,\,\,j=1,2,3\,.
\end{equation}
To find all the elements of the class $\mathcal S_s$, we have to solve
\begin{equation*}\label{E:2star}
\bigg\{\sum_{\ell=1}^{16} v_\ell(x)\,\Gamma_\ell,\alpha_j\bigg\}=0\quad \text{for all }j=1,2,3.
\end{equation*}
While we do not find here the general form of a scalar potential, \eqref{E:star} implies that $v_{12}\Gamma_{12}\in\cS_s$,
which in turn yields that $\cS_{Ls}$ is a strict subset of $\cS_s$. In addition, \eqref{E:star} also shows that the set of matrices 
$\{\alpha_1, \alpha_2, \alpha_3,\Gamma_{12}\}$ satisfies 
the same anticommutation relations \eqref{E:D.7} as the set $\{\alpha_1, \alpha_2, \alpha_3,\beta\}$. 
Since the proofs in Section~\ref{S:Dscal} only use these anticommutation relations (and not the concrete form of the matrices involved), 
we conclude that the criteria from Section~\ref{S:Dscal} apply to the operator $\ID_0+v_{12}\Gamma_{12}$.
Turning to dimension $d=4$, one can consider the following
$$
\ID_0=\sum_{j=1}^3 \alpha_j D_j+\Gamma_{12} D_4
$$
as the free Dirac operator, and note that, by \eqref{E:star}, $v_2\beta\in\cS_s$, meaning $\cS_s\not=\emptyset$.

Based on the discussion above, we formulate the following:\\
\textbf{Open problem:} Using the Clifford algebra formalism, write down the free Dirac operators in all dimensions $d$, 
identify the class of scalar potentials $\cS_s$, and develop the analogous theory to that of Section~\ref{S:Dscal}.

One final note: In the cases considered in Section~\ref{S:Dscal}, the optimal increase rate at the boundary of $\Omega$ for 
Lorentz scalar potentials is independent of the dimension $d$. 
It is a natural question whether this is a generally valid fact in all dimensions. First recall that this is indeed the case for 
Schr\"odinger operators (see \cite{NN1}). 
At the technical level, one reason for this is the fact that the best constant in the Hardy inequality is dimension-independent, 
and the same Hardy inequality enters our proofs in Section~\ref{S:Dscal}. Heuristically, the loss of essential self-adjointness 
is caused by the fact that the particle reaches the boundary $\partial\Omega$, 
and hence only the motion along the normal to $\partial\Omega$ matters. In addition, recall that, in many respects, Lorentz scalar 
potentials behave as positive potentials in the Schr\"odinger case (see e.g. \cite{MOBP,SV}). All of the above make it 
tempting to think that the optimal increase rate is indeed dimension-independent.

\textbf{6.} We now consider the case of the Dirac operator with an electric potential
(see \cite[Ch. 4.2.2]{Th}), that is
\begin{equation}\label{E:CO.2}
\bbV_e(x)=
\begin{cases}
v_e(x)\mathds 1_2\,,&\quad d=1,2\,,\\
v_e(x)\mathds 1_4\,,&\quad d=3\,.
\end{cases}
\end{equation}
Note that $\bbV_e$ is not a scalar potential, in the sense of Definition~\ref{D:ScaPo}. In
one dimension, we showed in Lemma~\ref{L:NES} that $\sigma_2 D+v_e\mathds 1$
is never essentially self-adjoint, i.e. electric potentials are not confining. Moreover,
since the force induced by an electric potential has opposite signs for electrons and positrons, 
respectively, we can infer on physical grounds that electric potentials
are not confining in any $d\leq3$. To the best of our knowledge, a proof
of this claim in $d=2$ and 3 does not yet exist, and this is a significant and interesting open problem.

One can, however, consider a Dirac operator with a Lorentz scalar and an electric potential; in $d=3$ this
leads to an overall potential
\begin{equation*}
\bbV=\beta v_s+\lambda_e v_e \mathds 1_4\,.
\end{equation*}
If $v_s$ is such that $\beta v_s$ is confining, then
condition \eqref{E:P.1} should hold for $|\lambda_e|$ small enough,
and hence Theorem~\ref{T:P} would allow one to conclude confinement for the
mixed potential $\bbV$. On the contrary, if $|\lambda_e|$ is sufficiently large, one would expect
deconfinement. 

While, as discussed above, we cannot prove this confinement/deconfinement
transition in $d=2,3$, the following example in $d=1$ substantiates the picture. Consider
the Dirac operator in 1 dimension:
\begin{equation*}
\ID=\sigma_2 D+ \sigma_3 v_s +v_e\mathds 1_2\quad\text{with }\cD om(\ID)=C_0^1\big((a,b);\IC^2\big)\,,
\end{equation*}
for some bounded interval $(a,b)\subset \IR$. Assume that, for $\delta(x)=\text{dist}\big(x,\{a,b\}\big)$
small enough, 
\begin{equation*}
v_s(x)=\frac{\lambda_s}{\delta(x)}\quad\text{and}\quad v_e(x)=\frac{\lambda_e}{\delta(x)}\,.
\end{equation*}
Then Corollary~\ref{C:SMF} states that $\ID$ is essentially self-adjoint if and only if
\begin{equation*}
\lambda_e^2\leq \lambda_s^2-\tfrac14\,.
\end{equation*}

\textbf{7.} Our last two comments are about magnetic confinement. First note that,
when we add rotationally invariant scalar and electric potentials to a (rotationally invariant)
magnetic potential $\bscA$ as in \eqref{E:M.19}, we can still perform partial wave
decomposition of
\begin{equation}\label{E:CO.3}
\ID=\bssig\cdot\big(\bsD-\bscA\big)+\sigma_3 v_s+v_e\mathds 1_2\quad\text{with }
\cD om(\ID)=C_0^1\big(\{|x|<1\};\IC^2\big)\,.
\end{equation} 
This toghether with Proposition  \ref{L:P} implies that $\ID$ is essentially self-adjoint if and only if 
\begin{equation}\label{E:CO.4}
\ID_{m_j}=\sigma_2 D_r+\sigma_1\big(a(r)-\tfrac{m_j}{r}\big)+\sigma_3 v_s(r)+v_e(r)\mathds 1_2
\end{equation}
is essentially self-adjoint on $\cD om(\ID_{m_j})=C_0^1\big((0,1);\IC^2\big)$ for each $j\in\bbZ$, 
where $D_r$ and $m_j$ are given in \eqref{E:M.23.1}.
If we take potentials such that, for $1-r$ small enough:
\begin{equation*}
a(r)=\frac{\lambda_m}{1-r}\,,\quad v_s(r)=\frac{\lambda_s}{1-r}\,,\quad\text{and}\quad 
v_e(r)=\frac{\lambda_e}{1-r}\,,
\end{equation*}
then Corollary~\ref{C:SMF} shows that, for any $j\in\bbZ$, $\ID_{m_j}$ is essentially
self-adjoint if and only if
\begin{equation}\label{E:CO.5}
\lambda_e^2\leq\lambda_m^2+\lambda_s^2-\tfrac14\,.
\end{equation}
In particular, we have shown that 
\begin{equation*}
\ID=\bssig\cdot(\bsD-\bscA)+v_e\mathds 1_2\quad\text{with }
\cD om(\ID)=C_0^1\big(\{|x|<1\};\IC^2\big)
\end{equation*}
and with, for $1-r$ small enough, magnetic field $\cB(r)=\frac{\lambda_m}{r(1-r)^2}$ and electric potential 
$v_e(r)=\frac{\lambda_e}{1-r}$ is essentially self-adjoint if and only if
\begin{equation}\label{E:CO.6}
\lambda_e^2\leq\lambda_m^2-\tfrac14\,.
\end{equation}
In other words, we find a confinement/deconfinement transition for purely electromagnetic potentials.
For the same type of results with confinement defined in a weaker sense, see \cite{MS}.

\textbf{8.} In our last comment we consider an example of magnetic confinement in three dimensions. 
As explained in the Introduction, a general theory
does not yet exist in this case, so we give an example where symmetries allow us to reduce the
problem to a 2-dimensional one. More precisely, let $\Om\subset\IR^2$
be as in Theorem~\ref{T:M2}, and set
\begin{equation*}
\tilde\Om=\Om\times\IR=\big\{x=(x_1,x_2,x_3)\,\big|\, (x_1,x_2)\in\Om\,,\,\,\, x_3\in\IR\big\}\,.
\end{equation*}
We further assume that the magnetic field doesn't depend on $x_3$, and that it is parallel
to the $x_3$-axis:
\begin{equation}\label{E:CO.7}
\bscB(x_1,x_2,x_3)=\big(0\,,\,\,0\,,\,\,\cB(x_1,x_2)\big)\,.
\end{equation}
In this case, we can choose (see the discussion around equation \eqref{E:M.15})
a magnetic potential of the form
\begin{equation*}
\bscA(x_1,x_2,x_3)=\big(\cA_1(x_1,x_2)\,,\,\,\cA_2(x_1,x_2)\,,\,\,0\big)\,,
\end{equation*}
where $\cA_1$ and $\cA_2$ are related to $\cB$ via \eqref{E:M.4}. We seek conditions
on $\cB$ ensuring essential self-adjointness of
\begin{equation}\label{E:CO.8}
\ID_{3,mag}=\alpha_1\,\big(D_1-\cA_1\big)+\alpha_2\,\big(D_2-\cA_2\big)+\alpha_3 D_3 
\end{equation}
on $\cD om(\ID_{3,mag})=C_0^2\big(\tilde\Om;\IC^4\big)$. From \eqref{E:D.5}, we note that
\begin{equation}\label{E:D.3m}
\ID_{3,mag}=
\left(\begin{matrix}
0 & \tilde \ID_{3,mag} \\
\tilde\ID_{3,mag} & 0 
\end{matrix}\right)
\end{equation}
where
\begin{equation*}
\tilde\ID_{3,mag}=\sigma_1\,\big(D_1-\cA_1\big)+\sigma_2\,\big(D_2-\cA_2\big)+\sigma_3 D_3 
\end{equation*}
with $\cD om(\tilde\ID_{3,mag})=C_0^2\big(\tilde\Om;\IC^2\big)$. A straightforward, abstract argument
shows that $\ID_{3,mag}$ is essentially self-adjoint iff $\tilde\ID_{3,mag}$ is essentially self-adjoint.

Consider the Fourier transform in $x_3$:
\begin{equation*}
\hat\bsPsi(x_1,x_2,\xi)=\big(\cF_3\bsPsi\big)(x_1,x_2,\xi)=\frac{1}{\sqrt{2\pi}}\int_\IR e^{-ix_3\xi}\bsPsi(x_1,x_2,x_3)\,dx_3\,,
\end{equation*}
using which we define
\begin{equation*}
\hat\ID_{3,mag}=\cF_3\tilde\ID_{3,mag}\cF_3^{-1}\,.
\end{equation*}
Let $\bsPsi\in C_0^2\big(\tilde\Om;\IC^2\big)$. One can easily check that, for each fixed $\xi \in\IR$, we have
$\hat\bsPsi(\cdot,\cdot,\xi)\in C_0^2\big(\Om;\IC^2\big)$ and the map
\begin{equation*}
C_0^2\big(\tilde\Om;\IC^2\big)\ni\bsPsi\,\,\mapsto\,\, \hat\bsPsi (\cdot,\cdot,\xi)\in C_0^2\big(\Om;\IC^2\big)
\end{equation*}
is surjective. Since $\bsPsi$ is assumed to have compact support, we also see that
$(x_1,x_2,\xi)\mapsto \xi\hat\bsPsi(x_1,x_2,\xi)\in L^2\big(\tilde\Om;\IC^2\big)$ and
$$
\hat\ID_{3,mag}=\sigma_1\,\big(D_1-\cA_1\big)+\sigma_2\,\big(D_2-\cA_2\big)+\sigma_3 \xi\,.
$$
It follows from all of the above that $\hat\ID_{3,mag}$ can be written as a direct integral:
$$
\hat\ID_{3,mag}=\int_\IR^\oplus \ID_\xi\,d\xi\quad\text{in }
L^2\big(\tilde\Om;\IC^2\big)=\int_\IR^\oplus L^2\big(\Om;\IC^2\big)\,d\xi
$$
with fiber operator
$$
\ID_\xi=\sigma_1\,\big(D_1-\cA_1\big)+\sigma_2\,\big(D_2-\cA_2\big)+\sigma_3 \xi \quad\text{on }
\cD om(\ID_\xi)=C_0^2\big(\Om;\IC^2\big)\,.
$$

Now assume that $\cB$ satisfies \eqref{E:TM2.1} or \eqref{E:TM2.2}. From Theorem~\ref{T:M2},
we know that $\ID_\xi$ is essentially self-adjoint for each $\xi\in\IR$. The essential self-adjointness of
$\hat\ID_{3,mag}$ (and hence $\tilde\ID_{3,mag}$ and $\ID_{3,mag}$) then follows by standard
arguments on orthogonal integrals (see, for example, \cite{RS}). Finally, we close by noting that,
while in this case we use the Fourier transform, this does not work if we are interested in general
cylinders of the form $\tilde\Om=\Om\times (a,b)$ with $a$ and/or $b$ finite. In such a case, 
one needs additional assumptions to ensure confinement in the $x_3$ direction, and
one can probably approach such a case using tensor products.


\end{document}